\documentclass[11pt]{article}

\usepackage[margin=1in]{geometry}
\usepackage{amsmath,amssymb,amsthm,mathtools,braket}
\usepackage{thmtools,thm-restate}

\usepackage{multirow}
\usepackage{microtype}
\usepackage{enumitem}
\usepackage{booktabs}
\usepackage{xcolor}
\usepackage{hyperref}
\usepackage{tikz}
\usetikzlibrary{arrows.meta}

\usepackage{authblk}
\usepackage[nameinlink,capitalise,noabbrev]{cleveref}
\hypersetup{
  hypertexnames=false, 
  colorlinks=true,
  linkcolor=blue!55!black,
  citecolor=green!45!black,
  urlcolor=blue!60!black
}

\newtheorem{theorem}{Theorem}[section]
\newtheorem{lemma}[theorem]{Lemma}

\newtheorem{corollary}[theorem]{Corollary}
\newtheorem{claim}[theorem]{Claim}
\theoremstyle{definition}
\newtheorem{definition}[theorem]{Definition}

\newtheorem{remark}[theorem]{Remark}
\newtheorem{fact}[theorem]{Fact}

\crefname{theorem}{Theorem}{Theorems}
\crefname{lemma}{Lemma}{Lemmas}
\crefname{proposition}{Proposition}{Propositions}
\crefname{corollary}{Corollary}{Corollaries}
\crefname{definition}{Definition}{Definitions}
\crefname{conjecture}{Conjecture}{Conjectures}
\crefname{remark}{Remark}{Remarks}
\crefname{fact}{Fact}{Facts}
\crefname{appendix}{Appendix}{Appendices}
\Crefname{appendix}{Appendix}{Appendices}

\newcommand{\C}{\mathbb C}
\newcommand{\R}{\mathbb R}
\DeclareMathOperator{\Gr}{Gr}
\DeclareMathOperator{\rank}{rank}
\DeclareMathOperator{\im}{im}
\newcommand{\op}{\mathrm{op}}
\DeclareMathOperator{\Tr}{Tr}
\newcommand{\Prob}{\mathrm{Pr}}
\newcommand{\E}{\mathbb E}
\newcommand{\cN}{\mathcal N}
\newcommand{\cQ}{\mathcal Q}
\newcommand{\eps}{\varepsilon}
\newcommand{\smin}{\sigma_{\min}}
\newcommand{\splus}{\sigma_{\min}^{+}}
\newcommand{\norm}[1]{\left\lVert #1\right\rVert}
\newcommand{\abs}[1]{\left\lvert #1\right\rvert}
\renewcommand{\set}[1]{\left\{#1\right\}}
\newcommand{\restr}[2]{#1\!\mid_{#2}}

\renewcommand{\L}{\mathsf{L}}
\newcommand{\FSLE}{\mathsf{FSLE}}
\newcommand{\NULL}{\mathsf{NULL}}
\newcommand{\SINGULAR}{\mathsf{SINGULAR}}
\newcommand{\pcFSLE}{\mathsf{pc}\text{-}\mathsf{FSLE}}
\newcommand{\pcNULL}{\mathsf{pc}\text{-}\mathsf{NULL}}
\newcommand{\pcSINGULAR}{\mathsf{pc}\text{-}\mathsf{SINGULAR}}
\newcommand{\pcFewSINGULAR}{\mathsf{pc}\text{-}\mathsf{FewSINGULAR}}
\newcommand{\pcUSINGULAR}{\mathsf{pc}\text{-}\mathsf{UniqueSINGULAR}}
\newcommand{\cNULL}{\text{-}\mathsf{conditioned}\text{-}\mathsf{NULL}}

\DeclareMathOperator{\poly}{poly}
\DeclareMathOperator{\ord}{ord}
\DeclareMathOperator{\dist}{dist}

\title{Robust subspace designs \\and the power of a unique small quantum witness}

\author[1]{Simon Apers}
\author[1]{Roman Edenhofer}
\author[1]{Benjamin Mathieu-Bloise}
\author[2]{Partha Mukhopadhyay}
\affil[1]{Université Paris Cité, CNRS, IRIF, Paris, France}
\affil[2]{Chennai Mathematical Institute, Chennai, India}

\begin{document}

\maketitle
\thispagestyle{empty}

\begin{abstract}
    The concept of subspace designs was introduced by Guruswami and Xing (STOC'13), and explicit constructions were given by Guruswami and Kopparty (FOCS'13). These are families of subspaces that have small intersection with any given subspace of a fixed dimension. We introduce \emph{robust subspace designs}. Informally, these are a quantitative extension in which we demand that not too many subspaces of the family contain directions that lie \emph{close} to any given subspace of a fixed dimension. We give a probabilistic construction of such a robust subspace design of polynomial size, as well as a non-trivial explicit construction of superpolynomial size.

    Our main application of this new concept is a quantum space-bounded variant of the Valiant-Vazirani theorem (Theor.~Comput.~Sci.'86), which shows that restricting $\mathsf{NP}$-complete problems to instances with at most one accepting witness preserves hardness under randomized reductions. For quantum witnesses, the analogous quantity is the dimension of an accepting witness subspace. We use our probabilistic construction of robust subspace designs to isolate a unique witness for space-bounded quantum Merlin-Arthur protocols with perfect completeness and an acceptance gap outside their perfectly accepting subspace.

    As a further application, we give a randomized reduction of well-conditioned nullity testing to space-bounded quantum Merlin-Arthur protocols with perfect completeness. Using a similar idea, we find that ordinary subspace designs allow us to recover the classical $\mathsf{C_= L}$ containment of Allender, Beals, and Ogihara (STOC'96) for general nullity testing through a simpler proof.
\end{abstract}

\clearpage
\thispagestyle{empty}
\tableofcontents
\clearpage
\pagenumbering{arabic}

\section{Introduction and summary}

\subsection{Robust subspace designs}

Subspace designs were introduced by Guruswami and Xing \cite{GuruswamiXing2013SubcodesSingleton}, who also gave a probabilistic polynomial-time construction. These are families of subspaces that have limited intersections with every subspace of a fixed dimension. Originally, the notion was defined over finite fields, mainly motivated by applications in coding theory. However, subspace designs can be defined over any field~\cite{BrakensiekChenDharZhang2026}.

In the weak version, a family of subspaces $\mathcal{H}=\{H_1,\ldots,H_m\}$ is a $(k,\Delta)$-weak subspace design if, for every $k$-dimensional subspace $W$,
\[
    \left|\{i \in [m]: H_i\cap W\neq \{0\}\}\right| \leq \Delta.
\]
The family is a $(k,\Delta)$-strong subspace design if, for every $k$-dimensional subspace $W$,
\[
    \sum_{i=1}^m \dim(H_i\cap W)\leq \Delta.
\]
Guruswami and Xing \cite{GuruswamiXing2013SubcodesSingleton} used subspace designs mainly in connection with list decoding. Subsequently, Guruswami and Kopparty explicitly constructed subspace designs with near-optimal parameters \cite{GuruswamiKopparty2016ExplicitSubspaceDesigns}. Since then, subspace designs have become useful objects in coding theory and algebraic pseudorandomness \cite{forbes2015dimension}. In a recent breakthrough, subspace designs were indirectly used to obtain the first $\mathsf{NC}$ algorithm for bipartite matching \cite{chatterjee2026bipartite}. Soon after that, they were also used to show that $\mathsf{NL}\subseteq \oplus\mathsf{L}$ \cite{Applebaum2026}.

In this work, we formulate and study a robust extension of this notion over the field of complex numbers $\mathbb C$.
For a linear map $Q:\C^n\to\C^k$ and a subspace $W$ of dimension $k$, it is based on the smallest singular value of $Q$ restricted to $W$,
\[
    \sigma_{\min}(Q|_W).
\]
If $Q$ is a co-isometry, then this is the sine of the smallest principal angle between $\ker(Q)$ and~$W$, or alternatively the minimum distance between any unit vector from $W$ to $\ker(Q)$.
Evidently, $\sigma_{\min}(Q|_W) = 0$ if and only if $\ker(Q) \cap W \neq 0$.

Now fix a robustness parameter $\eta \in [0,1]$.
We say that a family of linear maps $\mathcal{Q}=\{Q_1,\ldots,Q_m\}$ from $\C^n$ to $\C^k$, with bounded operator norms $\|Q_i\|_\op\leq 1$, is a~$(k,\eta,\Delta)$-weak robust subspace design if, for every $k$-dimensional subspace $W$,
\[
    \left|\{i\in[m] : \sigma_{\min}(Q_i|_W)\leq \eta\}\right|
    \leq \Delta.
\]

Similarly, the family~$\mathcal{Q}$ is a~$(k,\eta,\Delta)$-strong robust subspace design if, for every $k$-dimensional subspace $W$,
\[
    \sum_{i=1}^m \left|\{j\in[k] : \sigma_j(Q_i|_W)\leq\eta\}\right|
    \leq \Delta.
\]

It is straightforward to observe that if $\mathcal{Q}$ is a~$(k,\eta,\Delta)$-weak/strong robust subspace design, then the family of subspaces
\[
    \mathcal{H}=\{\ker(Q):Q\in\mathcal{Q}\}
\]
is an ordinary $(k,\Delta)$-weak/strong subspace design.
In some sense, the converse is not true: for every $\eta = \eta(n) \in(0,1]$, there is a $(1,n-1)$ subspace design whose robustness parameter is smaller than $\eta$.
Conceptually, this is illustrated in \cref{fig:robust-nonrobust}.

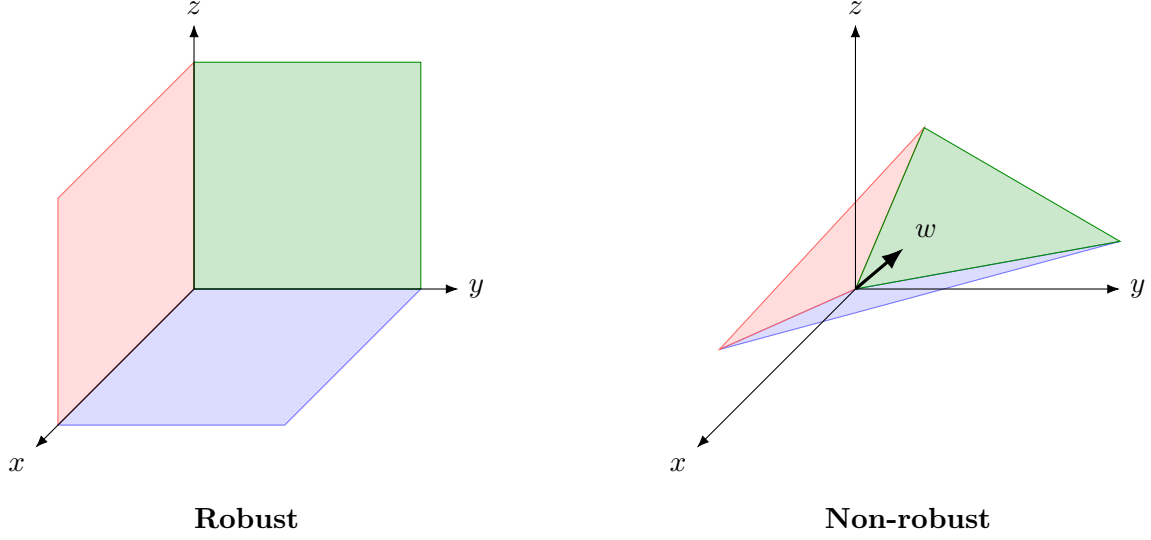
\begin{figure}[t]
    \centering

    \begin{minipage}[t]{0.47\textwidth}
        \centering
        \begin{tikzpicture}[>=Latex,scale=1,
            x={(-0.6cm,-0.6cm)}, y={(1cm,0cm)}, z={(0cm,1cm)}]
            \def\L{3}

            \filldraw[blue!55,fill opacity=.25]
            (0,0,0)--(\L,0,0)--(\L,\L,0)--(0,\L,0)--cycle;
            \filldraw[red!60,fill opacity=.22]
            (0,0,0)--(\L,0,0)--(\L,0,\L)--(0,0,\L)--cycle;
            \filldraw[green!55!black,fill opacity=.20]
            (0,0,0)--(0,\L,0)--(0,\L,\L)--(0,0,\L)--cycle;

            \draw[->] (0,0,0)--(3.5,0,0) node[below left] {$x$};
            \draw[->] (0,0,0)--(0,3.5,0) node[right] {$y$};
            \draw[->] (0,0,0)--(0,0,3.5) node[above] {$z$};
        \end{tikzpicture}

        \vspace{0.4em}
        \textbf{Robust}
    \end{minipage}
    \hfill
    \begin{minipage}[t]{0.47\textwidth}
        \centering
        \begin{tikzpicture}[>=Latex,scale=1,
            x={(-0.6cm,-0.6cm)}, y={(1cm,0cm)}, z={(0cm,1cm)}]
            \def\L{3.5}
            \def\t{.65}
            \def\a{.18}

            \coordinate (X) at (3,0,1);
            \coordinate (Y) at (0,\L,{\a*\L});
            \coordinate (C) at ({\t*\L},{\t*\L},\L);

            \filldraw[blue!55,fill opacity=.25]
            (0,0,0)--(X)--(Y)--cycle;
            \filldraw[red!60,fill opacity=.22]
            (0,0,0)--(X)--(C)--cycle;
            \filldraw[green!55!black,fill opacity=.20]
            (0,0,0)--(Y)--(C)--cycle;

            \draw[->] (0,0,0)--(3.5,0,0) node[below left] {$x$};
            \draw[->] (0,0,0)--(0,3.5,0) node[right] {$y$};
            \draw[->] (0,0,0)--(0,0,3.5) node[above] {$z$};
            \draw[very thick,black,->] (0,0,0)--(1.35,1.45,1.35) node[above right] {$w$};
        \end{tikzpicture}

        \vspace{0.4em}
        \textbf{Non-robust}
    \end{minipage}

    \caption{A schematic robust and non-robust family of planes in $\mathbb R^3$. No direction can be too close to all three coordinate planes on the left, whereas the direction $w$ is simultaneously close to all three planes on the right.}
    \label{fig:robust-nonrobust}
\end{figure}

Alongside applications of robust subspace designs in space-bounded complexity, we prove the following:
\begin{enumerate}
    \item We use random matrix theory to argue that, with constant probability, a polynomial-sized ensemble of Gaussian matrices forms a strong subspace design with inverse polynomial robustness.
    This implies a randomized-logspace-constructible robust subspace design.
    \item We combine ideas from low-rank matrix approximation (specifically, sketching and max-volume submatrices) to give a deterministic space-efficient weak robust subspace design with inverse polynomial robustness of size $O(n^{O(1)} k^{O(k)})$.
\end{enumerate}

\noindent
Before detailing these constructions in \cref{sec:intro-prob,sec:explicit-intro}, we first turn to our main application, which also serves as a concrete motivation for studying robust subspace designs.

\subsection{Isolating a unique small quantum witness}\label{sec:matrix-completion-logspace}
The Valiant-Vazirani theorem \cite{ValiantV86} shows that randomness can isolate a unique witness among several classical witnesses.
In complexity theoretic terms, it establishes $\mathsf{NP}\subseteq\mathsf{RP}^{\mathsf{promiseUP}}$.
Aharonov, Ben-Or, Brandão and Sattath \cite{Aharonov_2022} extended this result to bounded-error quantum verifiers with classical witnesses, showing that $\mathsf{QCMA} \subseteq \mathsf{RP^{UQCMA}}$.

In our work, we consider quantum witnesses instead.
In that case the perfectly accepting witnesses form a subspace, and their number is naturally measured by its dimension.

Now suppose that a quantum verifier's accepting subspace has dimension $k$, with witnesses on $w$ qubits.
Classically, if $k$ is at most polynomial in $n$, a prover can just send all $k$ witnesses in a canonical order to obtain a unique combined witness and let a polynomial time verifier check their distinctness and verify each one.
A natural quantum analogue is to have the prover send their antisymmetric product. For an orthonormal basis $\ket{\psi_1},\dots,\ket{\psi_k}$ of the accepting subspace, this is
\[
    \frac{1}{\sqrt{k!}}
    \sum_{\sigma\in S_k} \operatorname{sgn}(\sigma)\,
    \ket{\psi_{\sigma(1)}}\otimes\ket{\psi_{\sigma(2)}}\otimes
    \cdots\otimes\ket{\psi_{\sigma(k)}},
\]
where $S_k$ is the collection of all permutations $\sigma:[k]\rightarrow[k]$.
A quantum verifier can test whether a provided state lies in the antisymmetric subspace, i.e. whether it is spanned by antisymmetric products of $k$ orthogonal states (ensuring that the prover cannot send $k$ times the same witness state) and it can verify acceptance on each register. These conditions leave a one-dimensional subspace.
Jain, Kerenidis, Kuperberg, Santha, Sattath and Zhang~\cite{jain2011poweruniquequantumwitness} use this observation to reduce polynomially many quantum witnesses to a unique one, showing that~$\mathsf{FewQMA \subseteq P^{UQMA}}$.
This unique witness, however, occupies $k$ registers and therefore $k\cdot w$ qubits in total.

This raises a space-bounded version of the isolation question:
Can we obtain a unique quantum witness without asymptotically increasing the size of the witness register? We answer this question affirmatively for space-bounded unitary quantum verifiers with perfect completeness and an acceptance gap outside their perfectly accepting subspace.
It should be observed that for bounded-error quantum verifiers this is almost immediate, because, using space-efficient error reduction, they can actually remove the prover entirely by sampling from the maximally mixed state instead (at the cost of a runtime blowup), see \cite[Theorem 35]{fefferman_lin}.
In complexity-theoretic language, we prove:
\begin{restatable}{theorem}{qwitnessisolation}
    \label{thm:fewqumal1-in-rl-to-uqumal1}
    $\mathsf{FewQ_U MAL}_1 \subseteq \mathsf{RL}^{\mathsf{UQ_U MAL}_1}$.
\end{restatable}
\noindent Here $\mathsf{FewQ_U MAL}_1$ and $\mathsf{UQ_U MAL}_1$ are quantum Merlin-Arthur protocols on $O(\log n)$ qubits for the witness and auxiliary work register with perfect completeness and, on yes-instances, a  perfectly-accepting witness subspace of dimension $\poly(n)$ and $1$, respectively. Both require an acceptance gap on the orthogonal complement of the accepting subspace, that is, if a state is orthogonal to the subspace of states that are accepted with probability $1$, then it is accepted with probability at most $1/3$.
In the case of $\mathsf{FewQ_U MAL_1}$ the polynomial bound on the witness space already follows from the fact that the verifier acts on $O(\log n)$ qubits only. The acceptance gap is the actual non-trivial limitation.

Our primary tool used to establish this quantum witness isolation is what we call \emph{kernel peeling}, which relies on efficiently-constructible robust subspace designs:
\begin{lemma}[Implicit from \cref{lem:kernel-peeling,rem:psd-kernel-peeling}]
    Let $\eta,\gamma \in (0,1]$ and $k \geq k'$. Given a $(k',\eta,\Delta)$-weak robust subspace design $\mathcal{Q} = \{Q_1,\dots,Q_m\} \subseteq \mathbb C^{k'\times N}$ and a positive semi-definite matrix $H\in\mathbb C^{N \times N}$ with kernel dimension $k$ and smallest non-zero eigenvalue $\lambda_{\min}^+(H) \geq \gamma$, there exists an index $i\in[m]$ such that
    \[
        \dim\ker(H + Q_i^* Q_i) = k - k',
        \qquad
        \text{and}
        \qquad
        \lambda_{\min}^+(H+Q_i^*Q_i) \geq \eta^2\gamma/3.
    \]
\end{lemma}

To see how we use this, recall that by the standard circuit-to-Hamiltonian construction due to Kitaev \cite{KitaevShenVyalyi2002}, we find that for every $\mathsf{FewQ_U MAL_1}$ protocol for a length $n$ input $x$, there is an efficiently constructible positive semi-definite $\poly(n) \times \poly(n)$ matrix $H=H(x)$ whose kernel dimension is equal to the number of perfectly accepting witnesses.
Furthermore, from the acceptance gap on yes-instances and bounded acceptance probability on no-instances, one can derive an inverse polynomial spectral gap of the constructed Hamiltonian, $\lambda_{\min}^+(H) \geq 1/\poly(n)$.

If the witness subspace initially is of dimension $k=\dim\ker(H)$, then by taking a $(k-1,\eta,\Delta)$-weak robust subspace design $\{Q_1,...,Q_m\}$ with robustness parameter $\eta \geq 1/\poly(n)$, there will be some $Q_i$ so that adding this perturbation to $H$ the kernel dimension becomes
\[
    \dim\ker(H + Q_i^*Q_i) = 1,
    \qquad
    \text{and}
    \qquad
    \lambda_{\min}^+(H+Q_i^*Q_i) \geq 1/\poly(n).
\]

Using the phase estimation based singularity testing algorithm from \cite[Lemma 35]{FeffermanRemscrim2021}, we find that this yields a unique witness protocol with acceptance gap scaling polynomially with the spectral gap of the given matrix, where the prover precisely provides the isolated $O(\log n)$-qubit kernel state of $H+Q_i^*Q_i$.

Of course we do not know the dimension of the witness subspace in advance.
Our reduction relies on efficiently constructing $(k',\eta,\Delta)$-weak robust subspace designs for every $k'\leq 2^w = \poly(n)$ using randomness, and queries whether $H+Q_i^*Q_i$ has a unique kernel state using the just mentioned unique witness kernel verification procedure from \cite[Lemma 35]{FeffermanRemscrim2021}.
If any query answer returns yes, we accept.
Once we reach $k'=\dim\ker(H) - 1$, we will find some matrix $Q_i$ so that the kernel dimension of $H+Q_i^*Q_i$ is indeed one and which preserves the inverse polynomial spectral gap, yielding acceptance.
If on the other hand, there is no accepting witness, then $H$ will be regular to begin with and adding the positive semi-definite perturbations $Q_i^*Q_i$ will keep the matrix regular and only increase the spectral gap, so that no kernel verification procedure will give a positive answer for any $k'$.
The resulting protocol is in the complexity class $\mathsf{RL}^{\mathsf{UQ_U MAL}_1}$. 

\subsection{Further applications}
As further applications of subspace designs and kernel peeling, we find alternative simpler proofs of known classical containment results, as well as new quantum containment results for linear algebraic problems such as testing the nullity of a given matrix or the feasibility of a system of linear equations.

For our classical applications, we consider the nullity testing problem $\NULL_k$, which asks to decide whether the kernel dimension of a given matrix $A$ is at least $k$ or less than $k$.
With a similar kernel peeling idea as for our quantum witness isolation, we find that a $(k-1,\Delta)$ ``non-robust'' subspace design allows us to reduce the problem to $\NULL_1$ which corresponds to deciding whether a given matrix is singular, and which is known to be in the exact counting class $\mathsf{C_= L}$.
We then notice that the known Wronskian-based (non-robust) subspace design constructions, as presented in \cite{GuruswamiKopparty2016ExplicitSubspaceDesigns,BrakensiekChenDharZhang2026}, can be implemented in deterministic logspace.
Combining these ideas yields an alternative proof for $\NULL_k \in \mathsf{C_= L}$, originally proven in \cite{allender_complexity_1999} using a line of reductions including Berkowitz's algorithm \cite{BERKOWITZ1984147}.
By a reduction from linear system feasibility, $\mathsf{FSLE}$, to nullity testing (also from \cite{allender_complexity_1999}) this yields an alternative simpler reduction path towards $\mathsf{FSLE} \in \mathsf{L^{C_= L}}$.

In the quantum setting, we consider a well-conditioned version of the nullity testing problem, $\pcNULL_k$. Here we also wish to decide if the kernel dimension of a given matrix is at least $k$, but we are promised that if the kernel dimension is below $k$, then the $k$-th smallest singular value is at least inverse polynomial, $\sigma_k(A)\geq 1/\poly(n)$.
An analogous argument as in the classical case, together with a careful analysis of the robustness, shows that we can use our probabilistic robust subspace design construction to put $\pcNULL_k$ in $\mathsf{coRL^{Q_U MAL_1}}$.
Further, we find that a well-conditioned version of linear system feasibility, $\pcFSLE$ is contained in $\mathsf{BPL^{Q_U MAL_1}}$.
This extends our understanding of the complexity of well-conditioned linear algebra.

\cref{tab:linear-algebra-problems} gives a summary of these and related results about the complexity of fundamental linear algebra problems.

\begin{table}[ht]
    \centering
    {\renewcommand{\arraystretch}{1.8}
    \begin{tabular}{p{3.0cm}|p{4.0cm}|p{7.8cm}}
        \textbf{Problem} & \textbf{General} & \textbf{Poly-conditioned}
        \\\hline

        $\det(A)\stackrel{?}{=}0$
        &
        $\mathsf{C_{=}L}$-complete
        \cite{allender_complexity_1999}
        &
        $\mathsf{Q_{U}MAL}_1$-complete
        \cite{FeffermanRemscrim2021}
        \\\hline

        \multirow[c]{2}{=}{\raisebox{1.6ex}{$\mathrm{nullity}(A)\stackrel{?}{\geq} k$}}
        &
        \multirow{2}{4.0cm}{$\mathsf{C_{=}L}$-complete
        \cite{allender_complexity_1999}}
        &
        $\mathsf{Q_{U}MAL}_1$-complete for $k=O(1)$
        \;[\Cref{thm:pcnull-reduces-to-pcsingular-for-constant-k}]
        \\[-.1cm]
        &&
        $\in \mathsf{coRL^{Q_{U}MAL_1}}$
        \;[\Cref{thm:pc-null-in-corl-to-qumal1}]
        \\\hline

        Feasible $Ax \stackrel{?}{=} b$
        &
        $\mathsf{L}^{\mathsf{C_{=}L}}$-complete
        \cite{allender_complexity_1999}
        &
        $\in \mathsf{BPL^{Q_U MAL_1}}$
        \;[\Cref{thm:pcfsle-in-bpl-to-qumal1}]
        \\\hline

        Solve $Ax=b$
        &
        $\mathsf{DET}$-complete
        \cite{cook_taxonomy_1985}
        &
        $\mathsf{BQL}$-complete
        \cite{tashma,fefferman_lin,FeffermanRemscrim2021}
    \end{tabular}}
    \caption{Computational complexity of different linear algebra problems.
    We distinguish the general case from the poly-conditioned case, where the input instances are promised to satisfy appropriate well-conditionedness assumptions. Details are given in \Cref{sec:isolating-a-unique-small-quantum-witness-and-other-applications}.}
    \label{tab:linear-algebra-problems}
\end{table}

\subsection{Probabilistic construction of a robust subspace design} \label{sec:intro-prob}

We next summarize our probabilistic construction of a robust subspace design.
It is based on complex Gaussian matrices and uses various tools from Gaussian random matrix theory \cite{Forrester2010, hough2006determinantal, Szarek1998MetricEntropyHomogeneousSpaces, Vershynin2018HighDimensionalProbability}.

A first key fact that we use is that if $Q$ is a $k\times n$ complex Gaussian matrix and a subspace $W\in\Gr(k,n)$\footnote{Here $\Gr(k,n)$ is the complex Grassmannian, which is the set of all $k$-dimensional subspaces of $\C^n$.} is fixed, then the restriction $Q|_W$ of $Q$ to $W$ is a square Gaussian matrix.
Now, for a weak robust design, we need to control the probability that the least singular value of $Q|_W$ is small.
For a strong robust design, however, we must also control how many singular values of $Q|_W$ can be simultaneously close to zero.

To do so, we use that the singular values of a square complex Gaussian matrix have a well-understood joint distribution.
In particular, the expected number of singular values below $u/\sqrt k$ of a $k \times k$ Gaussian matrix is $O(u^2)$, and we can prove strong tail bounds for this number.
Applying these estimates independently to the maps in our family gives a strong bound for every fixed subspace $W$.
Finally, we use a union bound on a net over the complex Grassmannian $\Gr(k,n)$ so that the same family works for every $k$-dimensional subspace simultaneously.
This gives the following theorem.

\begin{theorem}[Informal version of \cref{thm:gaussian-strong-robust-design}]\label{intro:robust-gaussian}
    For every $1\le k< n$, there is a distribution over
    \[
        m=O\bigl(k(n-k)\log(nk)\bigr)
    \]
    maps $Q_i:\C^n\to\C^k$ such that, with high probability, $\{Q_i\}$ is a $(k,\eta,\Delta)$-strong robust subspace design for $\eta = \Theta(1/\sqrt{nk})$ and $\Delta = m/4$.
\end{theorem}

\noindent
We furthermore find that the above construction can be approximated to finite-precision by a randomized algorithm that uses only $O(\log n)$ workspace, by presenting a simple space-bounded Gaussian scalar sampler.

\subsection{Explicit construction of a robust subspace design} \label{sec:explicit-intro}

We next turn to explicit constructions of weak robust subspace designs.
We first give a simple construction based on coordinate projections.  For
$I\subseteq[n]$, let
$P_I:\C^n\to\C^{|I|}$
denote the projection onto the coordinates indexed by $I$.

The basic idea is already visible when $k=1$.  If $v\in\C^n$ is a unit
vector, then at least one coordinate of $v$ has magnitude at least
$1/\sqrt n$.  Thus, among the coordinate projections, there is always
one that is non-negligible on the line spanned by $v$; the uniform vector
shows that the scale~$1/\sqrt n$ is the correct one in general.
For a $k$-dimensional subspace, an analogous argument based on a carefully
chosen $k\times k$ coordinate minor yields robustness at scale
$1/\sqrt{nk}$, leading to the following construction.

\begin{restatable}[Projection-based robust subspace design]
    {theorem}{projectiondesign}
    \label{thm:coordinate}
    The family
    \[
        \mathcal P_k
        =
        \{P_I:I\subseteq[n],\;|I|=k\}
    \]
    is a $(k,\eta,\Delta)$-weak robust subspace design with $\eta=\frac{1}{2\sqrt{nk}}$ and $\Delta=\binom nk-1$.
\end{restatable}

Thus the robustness parameter is inverse polynomial, but the family has
size $\binom nk$, growing roughly as $n^k$.  To reduce the size, we first use a sketching idea to compress the ambient space using a small family of sign matrices.

More precisely, we construct an explicit family of linear maps
\[
    S:\C^n\to\C^r,
    \qquad r=\Theta(k^2),
\]
such that for every $k$-dimensional subspace $W$, many of the maps $S$
embed $W$ with inverse-polynomial minimum singular value.
The corresponding sketching analysis argues that this holds for any sample space with $4$-wise independent, unbiased signs.
We can thus use a polynomial-size, logspace-enumerable $4$-wise independent sample space
based on a construction of Alon--Babai--Itai~\cite{ABI86}.

The sketching step reduces the ambient dimension from $n$ to
$r=\Theta(k^2)$, but its output dimension is still larger than $k$.
We therefore compose each sketch $S$ with one of the coordinate
projections
\[
    P_I:\C^r\to\C^k,
    \qquad |I|=k.
\]
Since there are only
\[
    \binom{r}{k}=k^{O(k)}
\]
such projections at this point, this yields the following explicit construction.

\begin{restatable}[Explicit robust subspace design]
    {theorem}{sketchingdesign}
    \label{thm:sketching-subspace}
    For every $1\leq k\leq n$, there is an explicitly computable family of
    linear maps $\C^n\to\C^k$ forming a
    $(k,\eta,\Delta)$-weak robust subspace design with
    \[
        m=O\!\left(n^4k^{O(k)}\right),
        \qquad
        \eta=\Omega\!\left(\frac{1}{k^{3/2}\sqrt n}\right),
        \qquad
        \Delta<m.
    \]
    Moreover, the family can be enumerated in space $O(\log n + k \log k)$, and its entries have polynomial bit complexity.
\end{restatable}

In particular, when
\[
    k=O\!\left(\frac{\log n}{\log\log n}\right),
\]
we have $k^{O(k)}=n^{O(1)}$, and enumeration uses only $O(\log n)$ space.
Hence in this regime the construction gives a polynomial-size, deterministic logspace-enumerable robust subspace design.

\subsection{Open problems and future directions}

We finish the introduction by discussing some questions that our work evokes.

First and foremost is the question of an \emph{explicit} (e.g. logspace-uniform) \emph{robust subspace design} of polynomial size and inverse polynomial robustness.
While our probabilistic construction shows the non-explicit existence of such a design, our complexity-theoretic applications would benefit from a suitable derandomization.
As an example, by \cref{cor:robust-l} such a design would improve the containment $\mathsf{FewQ_U MAL}_1 \subseteq \mathsf{RL}^{\mathsf{UQ_U MAL}_1}$ to $\mathsf{FewQ_U MAL}_1 \subseteq \mathsf{L}^{\mathsf{UQ_U MAL}_1}$.

One promising approach to such a derandomization is to study the connection with dimension expanders \cite{lubotzky2008dimension} and quantum expanders \cite{ben2008quantum}, both of which are linear-algebraic variants of expander graphs.
Following the discussion e.g.~in \cite{li2025linear}, a quantum expander can be thought of as a ``robust'' version of a dimension expander.
While prior work has used subspace designs to design dimension expanders \cite{forbes2015dimension}, existing constructions of explicit quantum expanders may be useful to constructing robust subspace designs.

Further, we wonder whether our techniques may also be useful for quantum witness isolation in the time-bounded regime. A key challenge is to handle exponentially large witness spaces without losing efficiency or the spectral gap needed for verification.
The quantum-oracle barrier of Anshu, Haferkamp, Hwang, and Nguyen \cite{anshu_et_al:LIPIcs.ITCS.2026.10} also suggests that such an extension would need to exploit information beyond black-box access to the verifier.

Finally, we mention a natural \emph{geometric viewpoint} of subspace designs over $\C$.
For a fixed subspace $H\subseteq\C^n$ of dimension $n-k$, the set of $k$-dimensional subspaces $W$ satisfying
\[
    W\cap H\neq\{0\}
\]
forms a Schubert variety in the Grassmannian $\operatorname{Gr}(k,n)$. Thus, an ordinary subspace design controls how often $W$ lies on these bad varieties.
This viewpoint is exploited in recent work of Brakensiek, Chen, Dhar, and Zhang \cite{BrakensiekChenDharZhang2026}
who sharpen the Guruswami-Kopparty bound \cite{GuruswamiKopparty2016ExplicitSubspaceDesigns}, and use a Grassmannian/Schubert-calculus viewpoint to show that the improved bound is essentially optimal (over algebraically closed fields).
It is an interesting direction to explore this geometric viewpoint for robust subspace designs, which require that~$W$ stay quantitatively away from them.

\subsection{Paper overview}

After the preliminaries in \cref{sec:preliminaries}, we formally define robust subspace designs and introduce the kernel peeling technique in \cref{sec:robust-subspace-designs-and-kernel-peeling}.
In the same section, we recall the classical construction of subspace designs based on Wronskian evaluation maps which can have arbitrarily bad robustness parameter $\eta$.
In \cref{sec:probabilistic-construction}, we describe our probabilistic construction of robust subspace designs based on Gaussian matrices.
In \cref{sec:explicit-construction}, we present our explicit deterministic constructions.
Finally, in \cref{sec:isolating-a-unique-small-quantum-witness-and-other-applications}, we show our application to isolate a unique small quantum witness as well as the linear algebra applications to nullity testing and linear system feasibility.

\section{Preliminaries}\label{sec:preliminaries}

For every $n\in \mathbb N$, we write $[n]:=\{1,\dots,n\}$ for the first $n$ natural numbers.

Throughout, all vector spaces we consider are finite-dimensional over the complex numbers $\mathbb C$, and equipped with the standard inner product and Euclidean norm $\|\cdot\|$.
For a subspace $W\subseteq \mathbb C^n$, we write $W^\perp$ for its orthogonal complement and $\Pi_W$ for the orthogonal projector onto $W$.
We further write
\[
    \mathrm{dist}(x,W) := \min_{w\in W} \|w-x\| = \|\Pi_{W^\perp} x\|
\]
for the minimal distance of $x$ to $W$.
If $W$ has dimension $k$, then $N_W \in \mathbb C^{n\times k}$ denotes an arbitrary matrix whose columns form an orthonormal basis of $W$.
As such,
\[
    N_W^* N_W = I_k,
    \qquad
    N_WN_W^* = \Pi_W.
\]
We write $\Gr(k,n)$ for the complex Grassmannian of $k$-dimensional subspaces of $\mathbb C^n$.

For a linear map $A: \mathbb C^n \rightarrow \mathbb C^m$, we denote by $A^*$ its Hermitian adjoint, by $\ker(A)$ and $\im(A)$ its kernel and image, and by $\rank(A)$ its rank.
We freely identify a linear map by its matrix representation, where the underlying bases will always be clear from context or mentioned explicitly.
The operator and Frobenius norms are
\[
    \|A\|_{\op} := \sup_{x\in \mathbb C^n\backslash \{0\}}\frac{\|Ax\|}{\|x\|},
    \qquad
    \|A\|_F := \sqrt{\Tr A^*A} = \sqrt{\sum_{i,j} |A(i,j)|^2},
\]
and satisfy $\|A\|_\op \leq \|A\|_F$.

For a Hermitian matrix $H\in\mathbb C^{n \times n}$, we list its eigenvalues in nondecreasing order,
\[
    \lambda_1(H) \leq \cdots \leq \lambda_n(H).
\]
For $A\in\mathbb C^{m\times n}$, set $r:= \min\{m,n\}$. We likewise list the singular values of $A$ in non-decreasing order
\[
    \sigma_1(A) \leq \cdots \leq \sigma_r(A).
\]
Equivalently,
\[
    \sigma_j(A)^2
    =
    \lambda_{n-r+j}(A^*A)
    =
    \lambda_{m-r+j}(AA^*),
    \quad \forall j\in[r].
\]
By the Courant-Fischer theorem,
\[
    \sigma_j(A)
    =
    \sup_{\substack{U\subseteq \mathbb C^n\\\dim U = r-j+1}} \inf_{x\in U\backslash \{0\}} \frac{\|Ax\|}{\|x\|}.
\]
We write $\smin(A):=\sigma_1(A)$, $\sigma_{\max}(A):= \sigma_r(A) = \|A\|_\op$, and $\splus(A)$ for the smallest non-zero singular value of $A$, with the convention $\splus(0)=\infty$.
For Hermitian matrices $H$, we use the analogous notation and convention for $\lambda_{\min}(H), \lambda_{\max}(H)$ and $\lambda_{\min}^+(H)$.
We will also use the standard singular value perturbation bound
\[
    |\sigma_j(A) - \sigma_j(B)|
    \leq
    \|A-B\|_\op
\]
for matrices $A$ and $B$ of the same dimensions.

Finally, we record how singular values behave under restriction to a subspace.
Let \mbox{$Q\in \mathbb C^{d\times n}$} with $d\leq n$, let $W\subseteq \mathbb C^n$ have dimension $k\leq d$, and let $N_W \in \mathbb C^{n\times k}$ as above.
Since \mbox{$N_W: \mathbb C^k \rightarrow W$} is unitary, $QN_W$ is a matrix representation of $Q|_W$.
Extending $N_W$ to a unitary \mbox{$U=[N_W \; N_{W^\perp}] \in \mathbb C^{n\times n}$} gives
\[
    Q\Pi_W U = [QN_W \;0].
\]
Hence, for every $j\in [k]$,
\[
    \sigma_j(Q|_W)
    =
    \sigma_j(QN_W)
    =
    \sigma_{d-k+j}(Q\Pi_W).
\]

Throughout, computational inputs are matrices over $\mathbb Q(i)$ with polynomial-bit entries. Arithmetic outputs are streamed bit by bit to a write-only unidirectional output tape by machines that may use less space than the output size.

\section{Robust subspace designs and kernel peeling}\label{sec:robust-subspace-designs-and-kernel-peeling}

We first recall the definition of weak and strong robust subspace designs from the introduction.
After that, we show that robust subspace designs give rise to our kernel peeling technique, the main technical ingredient for our applications.
Finally, we remark that the ordinary subspace design constructions based on Wronskian evaluation maps from \cite{GuruswamiKopparty2016ExplicitSubspaceDesigns,BrakensiekChenDharZhang2026} can be implemented space-efficiently.
However, we find they may have arbitrarily small robustness parameter~$\eta\in(0,1]$, making them useful only to our classical complexity applications, and not for quantum.

\subsection{Definitions}

For all definitions, let $\eta\in[0,1],n\in \mathbb N$ and $k\leq n$, $\Delta\geq 0$.

\begin{definition}[Weak robust subspace design]
    A family of linear maps \mbox{$Q_1,\ldots,Q_m:\mathbb C^n \rightarrow \mathbb C^{k}$} is a $(k,\eta,\Delta)$-weak robust subspace design if $\Delta < m$, $\norm{Q_i}_{\op}\le 1$ for every $i$, and for every $k$-dimensional subspace $W\subseteq\C^n$,
    \[
        \left|\{i\in[m] : \sigma_{\min}(Q_i|_W)\leq\eta\}\right| \leq \Delta.
    \]
\end{definition}

Informally, we call a map $Q_i$ \emph{bad} for $W$ if it contributes to the above count, i.e. if $\smin(Q_i|_W) \leq \eta$.

\begin{definition}[Strong robust subspace design]
    A family of linear maps \mbox{$Q_1,\ldots,Q_m:\mathbb C^n \rightarrow \mathbb C^{k}$} is a $(k,\eta,\Delta)$-strong robust subspace design if $\Delta < mk$, $\norm{Q_i}_{\op}\le 1$ for every $i$, and for every $k$-dimensional subspace $W\subseteq\C^n$,
    \[
        \sum_{i=1}^m \left|\{j\in[k] : \sigma_{j}(Q_i|_W)\leq\eta\}\right| \leq \Delta.
    \]
\end{definition}

\begin{remark}[Co-domain dimension condition]
    We restrict the linear maps in a robust subspace design to have co-domain of dimension $k$, because this is what we require for our linear algebraic applications.
    We do later consider the relaxation to larger co-domain maps, but then we call the family a sketching family (see \cref{def:sketching-family}) because of their relation to sketching.
\end{remark}

As in the non-robust world: Every strong robust design with $\Delta<m$ is also a weak robust design with the same parameters, because every bad map contributes at least one to the above sum.
Further, as noted in the introduction, it is straightforward to observe that
\[
    \mathcal{H} = \{\ker(Q_i): i \in [m]\}
\]
is a $(k,\Delta)$-weak/strong subspace design if the underlying $Q_i$'s are a $(k,\eta,\Delta)$-weak/strong robust subspace design for any $\eta \in [0,1]$.
Finally, we will sometimes also refer to non-robust subspace designs as families of maps, rather than as families of subspaces, by which we will mean the corresponding kernel-families.

\begin{remark}[Geometric interpretation]
    \label{rem:geometric-robust-subspace-design}
    Robust subspace designs admit a natural geometric interpretation in terms of approximate intersections of subspaces. For subspaces $W,H\subseteq\C^n$, define the $\eta$-\emph{robust intersection dimension} by
    \[
        \operatorname{rdim}_\eta(W,H)
        :=
        \max\Bigl\{
        \dim V :
        V\subseteq W,\;
        \dist(v,H)\leq \eta\norm{v}
        \text{ for all }v\in V
        \Bigr\}.
    \]
    Thus $\operatorname{rdim}_0(W,H)=\dim(W\cap H)$, while for $\eta>0$ the quantity also counts directions in $W$ that are close to lying in $H$.

    This notion is directly characterized by the singular values of the restricted projection. Indeed, if $\dim W=k$, then the Courant-Fischer min-max characterization of singular values gives
    \[
        \operatorname{rdim}_\eta(W,H)
        =
        \left|
        \left\{
        j\in[k]:
        \sigma_j\!\left(\Pi_{H^\perp}|_W\right)\leq\eta
        \right\}
        \right|.
    \]

    More generally, let $Q:\C^n\to\C^k$ satisfy $\norm{Q}_{\op}\leq1$ and let $H=\ker(Q)$. Since $Qx=Q\Pi_{H^\perp}x$, we have
    \[
        \norm{Qv}\leq\norm{\Pi_{H^\perp}v} = \mathrm{dist}(v,H)
        \qquad\text{for every }v\in\C^n.
    \]
    Consequently,
    \[
        \sigma_j(Q|_W)
        \leq
        \sigma_j(\Pi_{H^\perp}|_W)
        \qquad\text{for every }j\in[k].
    \]
    Hence a large robust intersection between $W$ and $\ker(Q)$ necessarily produces at least as many small singular values of $Q|_W$.

    When $Q$ is a co-isometry, i.e. in particular $Q^*Q=\Pi_{\ker(Q)^\perp}$, then the above inequality becomes an equality.
    Thus, in the special case where $Q$ is a co-isometry, our definition of a robust subspace design is exactly the geometric notion obtained by replacing ordinary intersections by robust intersections.
\end{remark}

Although this geometric viewpoint is intuitive, the formulation in terms of general linear maps is more flexible and better suited to our applications.

\subsection{Kernel peeling}\label{sec:robust}

In this subsection we show that an explicit (robust) subspace design allows us to apply a kernel peeling procedure, that is it allows us to compute a family of rank $k'$ perturbations so that for every matrix $A$ with kernel dimension $k\geq k'$, there is some perturbation in the family so that adding it to $A$ reduces its kernel dimension by exactly $k'$, and preserves a spectral gap.

Our main technical tool is the following lemma which states that if two maps $Q_R$ and $Q_L$ are invertible on the right- and left-kernel of $A$, respectively, then $A+Q_L^*Q_R$ must be invertible on the full space.
The lemma also proves a lower bound on the resulting minimal singular value of~$A+Q_L^*Q_R$ depending on the smallest non-zero singular value of $A$.

\begin{lemma}[Kernel removal]\label{lem:kernel-removal}
    Let $\eta,\gamma\in(0,1], 1\leq k \leq n$ and let $A\in\C^{n\times n}$ satisfy
    \[
        \dim\ker A=k,
        \qquad \splus(A)\ge\gamma.
    \]
    Further, let $Q_R,Q_L\in\C^{k\times n}$ satisfy $\|Q_R\|_{\op},\|Q_L\|_{\op}\le1$ and
    \[
        \smin\!\left(\restr{Q_R}{\ker A}\right),\;
        \smin\!\left(\restr{Q_L}{\ker A^*}\right)\geq\eta.
    \]
    Then $A+Q_L^*Q_R$ is invertible and\vspace{-.4cm}
    \[
        \smin(A+Q_L^*Q_R)\ge \frac{\eta^2\gamma}{3}.
    \]
\end{lemma}
\begin{proof}
    If $k=n$, then $A=0$ and $\smin(Q_L^*Q_R)\geq\eta^2\geq\eta^2\gamma/3$, so assume $k<n$.
    With respect to orthonormal bases adapted to
    \[
        \mathbb C^n = \ker(A)^\perp \oplus \ker(A)
        \qquad\text{and}\qquad
        \mathbb{C}^n = \ker(A^*)^\perp \oplus \ker(A^*),
    \]
    we find unitarily equivalent representations of $A$ and $B:=Q_L^*Q_R$ as
    \[
        A \sim
        \begin{pmatrix}
            A_{11}  & 0 \\
            0       & 0
        \end{pmatrix},
        \qquad
        B \sim
        \begin{pmatrix}
            B_{11}  & B_{12} \\
            B_{21}  & B_{22}
        \end{pmatrix},
    \]
    where $\splus(A) = \smin(A_{11}) \geq \gamma$ and $B_{22} = \left(\restr{Q_L}{\ker A^*}\right)^* \left(\restr{Q_R}{\ker A}\right)$ such that $\smin(B_{22}) \geq \eta^2$.
    In particular $B_{22}$ is invertible with $\rank(B_{22}) = k$.
    However, since $\rank(B)=k$ already, we find that we can express the first $n-k$ columns of the above block representation of $B$ by linear combinations of the last $k$ columns, i.e. there exists some matrix $X\in\mathbb C^{k\times (n-k)}$ such that
    \[
        \begin{pmatrix}
            B_{11}\\
            B_{21}
        \end{pmatrix}
        =
        \begin{pmatrix}
            B_{12}\\
            B_{22}
        \end{pmatrix}X.
    \]
    From the lower block we obtain that $X=B_{22}^{-1}B_{21}$ and then from the top block we get
    \begin{equation}\label{eq:schur-complement-vanishes}
        B_{11} = B_{12}B_{22}^{-1}B_{21}.
    \end{equation}
    Moreover,
    \begin{equation}\label{eq:op-bounds}
        \|B_{22}^{-1}B_{21}\|_{\op}, \;
        \|B_{12}B_{22}^{-1}\|_{\op} \leq 1/\eta,
    \end{equation}
    because, respectively,
    \[
        B_{22}^{-1}B_{21}
        = \left(\restr{Q_R}{\ker A}\right)^{-1}\restr{Q_R}{\ker(A)^\perp},
        \qquad
        B_{12}B_{22}^{-1}
        = \left(\restr{Q_L}{\ker(A^*)^\perp}\right)^*\left(\left(\restr{Q_L}{\ker A^*}\right)^*\right)^{-1}.
    \]
    Using \cref{eq:schur-complement-vanishes}, we find the following block factorization of $A+B$:
    \[
        A+B \sim
        \begin{pmatrix}
            A_{11} + B_{11} & B_{12}\\
            B_{21}          & B_{22}
        \end{pmatrix}
        =
        \begin{pmatrix}
            I   & B_{12}B_{22}^{-1}\\
            0   & I
        \end{pmatrix}
        \begin{pmatrix}
            A_{11}  & 0\\
            0       & B_{22}
        \end{pmatrix}
        \begin{pmatrix}
            I                   & 0\\
            B_{22}^{-1}B_{21}   &I
        \end{pmatrix}.
    \]
    As each block triangular matrix on the factorization on the right is regular, we find that also $A+B$ must be regular.
    A simple calculation verifies that its inverse can be factorized as
    \begin{align*}
        (A+B)^{-1} \sim&
        \begin{pmatrix}
            I&0\\
            -B_{22}^{-1}B_{21}&I
        \end{pmatrix}
        \begin{pmatrix}
            A_{11}^{-1}&0\\
            0&B_{22}^{-1}
        \end{pmatrix}
        \begin{pmatrix}
            I&-B_{12}B_{22}^{-1}\\
            0&I
        \end{pmatrix}\\
        &=
        \begin{pmatrix}I\\-B_{22}^{-1}B_{21}\end{pmatrix}
        A_{11}^{-1}
        \begin{pmatrix}I&-B_{12}B_{22}^{-1}\end{pmatrix}
        +
        \begin{pmatrix}0\\I\end{pmatrix}
        B_{22}^{-1}
        \begin{pmatrix}0&I\end{pmatrix}.
    \end{align*}
    Since the block matrix representation of $(A+B)^{-1}$ above is unitarily equivalent to $(A+B)^{-1}$, we can bound its operator norm by
    \begin{align*}
        \|(A+B)^{-1}\|_{\op}
        &\leq \sqrt{1+\|B_{22}^{-1}B_{21}\|_{\op}^2} \;\|A_{11}^{-1}\|_{\op} \sqrt{1+\|B_{12}B_{22}^{-1}\|_{\op}^2} + \|B_{22}^{-1}\|_{\op}\\
        &\leq \frac{1+\eta^{-2}}{\gamma} + \frac{1}{\eta^2}
        \leq \frac{3}{\eta^2\gamma},
    \end{align*}
    where we used \cref{eq:op-bounds} to establish the second inequality.
    This gives
    \[
        \smin(A+B) = \frac{1}{\|(A+B)^{-1}\|_{\op}} \geq \frac{\eta^2 \gamma}{3}
    \]
    and finishes the proof.
\end{proof}

\begin{lemma}[Kernel peeling]\label{lem:kernel-peeling}
    Let $1\leq k'\leq k\leq n$, $\eta,\gamma \in (0,1]$, and let $A\in \mathbb{C}^{n\times n}$ satisfy
    \[
        \dim\ker A=k,
        \qquad \splus(A)\geq\gamma.
    \]
    Let $Q_1,\dots,Q_m : \mathbb C^n \to \mathbb C^{k'}$ be a $(k',\eta,\Delta)$-weak robust subspace design.
    Then there exist $i,j\in [m]$ such that\vspace{-.2cm}
    \[
        \dim \ker(A+Q_i^*Q_j) = k-k',
        \qquad \splus(A+Q_i^*Q_j) \geq \frac{\eta^2\gamma}{3}.
    \]
\end{lemma}
\begin{proof}
    Fix any $k'$-dimensional subspaces $W_R\subseteq \ker A$ and $W_L\subseteq \ker A^*$.
    Since $m>\Delta$, by the design property, there must exist indices $i,j\in[m]$ for which
    \[
        \smin(Q_j|_{W_R}), \;\smin(Q_i|_{W_L}) \geq \eta.\vspace{-.3cm}
    \]
    We set\vspace{-.2cm}
    \[
        W_R' := \ker(Q_j|_{\ker(A)}) = \ker(Q_j) \cap \ker(A),
        \qquad
        W_L' := \ker(Q_i|_{\ker(A^*)}) = \ker(Q_i) \cap \ker(A^*),
    \]
    so that $W_R' + W_R = \ker(A)$ and $W_L'+W_L = \ker(A^*)$.
    Since $\dim \ker A = \dim \ker A^* = k$ and both $Q_j|_{\ker(A)}$ and $Q_i|_{\ker(A^*)}$ have rank $k'$, the rank-nullity theorem yields
    \[
        \dim W_R' = \dim W_L' = k-k'.
    \]
    Consider the restrictions
    \[
        \widetilde A:=A|_{(W_R')^\perp}:(W_R')^\perp\to(W_L')^\perp,
        \qquad
        \widetilde Q_j:=Q_j|_{(W_R')^\perp},
        \qquad
        \widetilde Q_i:=Q_i|_{(W_L')^\perp}.
    \]
    Choosing orthonormal bases of $(W_R')^\perp$ and $(W_L')^\perp$, we identify these maps with matrices of size $\widetilde n \times \widetilde n$ and $k'\times \widetilde n$, respectively, where $\widetilde n:=n-k+k'$.
    Since $W_R'\subseteq\ker A$ and $W_L'\subseteq\ker A^*$,
    \[
        \dim\ker\widetilde A=k',
        \qquad
        \splus(\widetilde A)=\splus(A)\geq\gamma.
    \]
    Moreover,
    \[
        \ker\widetilde A
        =
        \ker A\cap(W_R')^\perp.
    \]
    As $W_R'=\ker(Q_j|_{\ker A})$, the restriction of $Q_j|_{\ker A}$ to
    $\ker A\cap(W_R')^\perp$ has precisely the nonzero singular values of
    $Q_j|_{\ker A}$. Therefore
    \[
        \smin(\widetilde Q_j|_{\ker\widetilde A})
        =
        \splus(Q_j|_{\ker A})
        \geq
        \smin(Q_j|_{W_R})
        \geq\eta,
    \]
    where the middle inequality follows from the Courant-Fischer min-max characterization of singular values. Similarly,
    \[
        \smin(\widetilde Q_i|_{\ker\widetilde A^*})\geq\eta.
    \]
    Hence \cref{lem:kernel-removal} gives
    \[
        \smin(\widetilde A+\widetilde Q_i^*\widetilde Q_j)
        \geq\frac{\eta^2\gamma}{3}.
    \]

    Finally, $A+Q_i^*Q_j$ vanishes on $W_R'$, and $(A+Q_i^*Q_j)^*$ vanishes on $W_L'$.
    Thus, its restriction from $(W_R')^\perp$ to $(W_L')^\perp$ is exactly $\widetilde A+\widetilde Q_i^*\widetilde Q_j$, which is invertible.
    Consequently,
    \[
        \ker(A+Q_i^*Q_j)=W_R',
    \]
    and thus
    \[
        \dim\ker(A+Q_i^*Q_j)=k-k',
        \qquad
        \splus(A+Q_i^*Q_j)\geq\frac{\eta^2\gamma}{3}.\qedhere
    \]
\end{proof}

Observe that the condition in \cref{lem:kernel-peeling} that the subspace design should be of size $m>\Delta$ simply ensures that the constraint of being a subspace design is non-trivial.

In most of our applications we will only need the above kernel peeling lemma for Hermitian matrices $A$. We make the following remark about this case:
\begin{remark}[Positive semidefinite perturbations]
    \label{rem:psd-kernel-peeling}
    If $A$ in \cref{lem:kernel-peeling} is Hermitian, then the two indices can be chosen equal: there exists $i\in[m]$ such that
    \[
        \dim\ker(A+Q_i^*Q_i)=k-k',
        \qquad
        \splus(A+Q_i^*Q_i)\geq\frac{\eta^2\gamma}{3}.
    \]
    Indeed, since $\ker(A)=\ker(A^*)$, we may choose $W_R=W_L$ in the proof of \cref{lem:kernel-peeling}.
    The condition $m>\Delta$ guarantees an index $i$ with $\smin(Q_i|_{W_R})\geq\eta$, and the remainder of the proof applies unchanged.
    Thus it suffices to consider the $m$ positive semi-definite perturbations $Q_i^*Q_i$ rather than all $m^2$ products $Q_i^*Q_j$.

    For general $A$, the same conclusion holds if $m>2\Delta$: at most $\Delta$ indices fail the required bound on each of $W_R$ and $W_L$, so some index satisfies both bounds.
\end{remark}

\subsection{Explicit non-robust construction based on Wronskian evaluation maps}

We first describe a simplified variant over $\C$ of the explicit subspace design due to Guruswami and Kopparty \cite{GuruswamiKopparty2016ExplicitSubspaceDesigns} that fits into our framework.
While their construction is based on the folded Wronskian over finite fields, it was extended to arbitrary fields and their parameter analysis was improved by Brakensiek, Chen, Dhar and Zhang \cite{BrakensiekChenDharZhang2026}.
We give a construction based on the standard Wronskian and mostly follow the analysis of the latter work.
We observe that the construction can be performed space-efficiently.
Since our analysis is similar to what is already known from \cite{BrakensiekChenDharZhang2026}, we defer proofs to the appendix.

For $x=(x_0,\ldots,x_{n-1})\in\C^n$, let
\[
    f_x(z) :=x_0+x_1z+\cdots+x_{n-1}z^{n-1}.
\]
For $1\le k\le n$ and $a\in\C$, define the linear map
\[
    Q_{k,a}:\C^n\longrightarrow\C^k,
    \qquad
    Q_{k,a}(x)=
    \bigl(f_x(a),f_x'(a),\ldots,f_x^{(k-1)}(a)\bigr)^T
\]
and note that $\rank(Q_{k,a})=k$.

\begin{restatable}[Wronskian subspace design]{theorem}{wronskiandesign}
    \label{thm:wronskian-subspace-design}
    For every $1\leq k <n$, and every set $T_k\subseteq \mathbb R$ of $m>k(n-k)$ points, the family \vspace{-.2cm}
    \[
        \mathcal H_{n,k} := \{H_a:a\in T_k\},
        \qquad\text{where}\qquad
        H_a:=\ker(Q_{k,a}),
    \]
    is a $(k,\;k(n-k))$-strong subspace design.
    Moreover, if $T_k=\{0,1,\dots,m-1\}$, then the entries of all $Q_{k,a}$ are of polynomial bit complexity and computable in space $O(\log nm)$.
\end{restatable}
We include a self-contained proof of the above theorem in \cref{sec:wronskian-subspace}.

\subsection{Subspace designs need not be robust}\label{sec:exactvsrobust}

Recall again that every robust subspace design induces an ordinary subspace design by taking the kernels of its maps. In this subsection, we show that the converse fails in a quantitative sense: An ordinary strong subspace design may have arbitrarily poor robustness.
In fact our example is a special case of the Wronskian design from \cref{thm:wronskian-subspace-design}.
This result qualitatively matches an argument in \cite[Theorem 1.8]{li2025linear} showing that dimension expanders can be arbitrarily poor quantum expanders.

\begin{theorem}\label{thm:subspace-design-need-not-be-robust}
    For every $n\geq2$ and every $\eta=\eta(n)\in(0,1]$, there exist maps
    \[
        Q_1,\ldots,Q_n:\C^n\to\C
    \]
    with $\|Q_i\|_{\op}= 1$ such that $\{\ker(Q_1),\ldots,\ker(Q_n)\}$ is a strong $(1,n-1)$-subspace design, while~$\{Q_1,\ldots,Q_n\}$ is not a $(1,\eta,n-1)$-weak robust subspace design.
\end{theorem}

\begin{proof}
    For $a\in\C$, consider
    \[
        Q_{1,a}(x_0,\ldots,x_{n-1})
        =
        x_0+x_1a+\cdots+x_{n-1}a^{n-1},
        \qquad
        \widehat Q_a:=\frac{Q_{1,a}}{\|Q_{1,a}\|_{\op}}.
    \]
    Choose $\varepsilon>0$ such that
    \[
        (n-1)\varepsilon <\eta
        \qquad\text{and let}\qquad
        T:=\{0,\varepsilon,2\varepsilon,\ldots,(n-1)\varepsilon\}.
    \]

    While it follows already from \cref{thm:wronskian-subspace-design}, we show that $\{\ker(\widehat Q_{a}):a\in T\}$ is a $(1,n-1)$ strong subspace design. To this end, let $W=\operatorname{span}\{y\}$ be any one-dimensional subspace of $\C^n$, where $y\neq0$, and associate with $y$ the nonzero polynomial
    \[
        f_y(z)=y_0+y_1z+\cdots+y_{n-1}z^{n-1}.
    \]
    Since
    \[
        \ker(\widehat Q_{a})\cap W\neq\{0\}
        \quad\Longleftrightarrow\quad
        f_y(a)=0,
    \]
    and $f_y$ has degree at most $n-1$, this can occur for at most $n-1$ points $a\in T$. For one-dimensional $W$, weak and strong subspace designs coincide, proving the claim.

    On the other hand, let
    \[
        W:=\operatorname{span}\{e_2\},
        \qquad
        e_2=(0,1,0,\ldots,0)^T.
    \]
    For every $a\in T$,
    \[
        \smin(\widehat Q_{a}|_W)
        =
        |\widehat Q_{a} e_2|
        =
        \frac{|a|}{\|Q_{1,a}\|_{\op}}
        \leq |a|
        \leq (n-1)\varepsilon
        <\eta,
    \]
    where we used that $\|Q_{1,a}\|_{\op}\geq1$. Thus all $n$ maps are bad on the same one-dimensional subspace $W$, and the family $\{\widehat Q_a:a\in T\}$ is not a $(1,\eta,n-1)$-weak robust subspace design.
\end{proof}

\section{Probabilistic robust construction}\label{sec:probabilistic-construction}

In this section we detail our randomized construction of robust subspace designs.
Our analysis borrows tools related to the Grassmannian and random matrix theory, which we introduce first.

\subsection{The Grassmannian manifold of subspaces}

Recall that we let $\Gr(k,n)$ denote the complex Grassmannian of $k$-dimensional subspaces of
$\C^n$.
We equip $\Gr(k,n)$ with the projection metric
\[
    d(W,V):=\norm{\Pi_W-\Pi_V}_{\op}.
\]
We will use a net with respect to this metric.
The following lemma follows from \cite{Szarek1998MetricEntropyHomogeneousSpaces}, and is stated in e.g.~\cite{mangoubi2022private}.

\begin{lemma}[Grassmannian net]
    \label{lem:net}
    There exists an absolute constant $C_{\mathrm{net}}>0$ such that, for every $0<\rho<1$, the complex Grassmannian $\Gr(k,n)$ admits a $\rho$-net $\{V_1,\ldots,V_M\}$ with respect to the projection metric, of cardinality
    \[
        M \leq \left(\frac{C_{\mathrm{net}}}{\rho}\right)^{2k(n-k)}.
    \]
    That is, for every $W\in\Gr(k,n)$, there exists $i\in[M]$ such that
    \[
        d(V_i,W)\leq \rho.
    \]
\end{lemma}

Finally, we will also use that the singular values of a projected matrix are Lipschitz continuous on the Grassmannian manifold.

\begin{lemma}[Lipschitz continuity]\label{lem:lipschitz}
    Let $Q:\C^n\to\C^k$ and $W,V\in\Gr(k,n)$.  For every $j\in[k]$,
    \[
        \abs{\sigma_j(\restr QW)-\sigma_j(\restr QV)}
        \leq \norm Q_{\op}\,\norm{\Pi_W-\Pi_V}_{\op}.
    \]
\end{lemma}
\begin{proof}
    Since $Q$ has a co-domain of dimension $k\leq n$ and $\dim W=\dim V=k$, we have $\sigma_j(Q|_W) = \sigma_j(Q\Pi_W)$ for every $j\leq k$, and similarly for $V$.
    Hence Weyl's inequality gives
    \[
        |\sigma_j(Q|_W)-\sigma_j(Q|_V)|
        = |\sigma_{j}(Q\Pi_W)-\sigma_{j}(Q\Pi_V)|
        \leq \|Q(\Pi_W-\Pi_V)\|_{\op},
    \]
    which is at most
    \[
        \|Q\|_{\op}\|\Pi_W-\Pi_V\|_{\op}.\qedhere
    \]
\end{proof}

\subsection{Random Gaussian matrix theory}

A standard complex Gaussian random variable, denoted $g\sim\mathcal N_\C(0,1)$,
has density $\pi^{-1}e^{-\abs z^2}$ on $\C$.  Equivalently,
$g=(X+iY)/\sqrt2$, where $X,Y$ are independent real standard Gaussians.
The following lemma shows concentration of the operator norm of $G$.
\begin{lemma}[{Gaussian operator norm \cite[Proposition 6.33]{aubrun2017alice}}]\label{lem:gaussian-op}
    Let $G\in\C^{k\times n}$, with $k\le n$, have independent entries distributed as
    $\mathcal N_\C(0,1/n)$.  Then
    \[
        \Prob[\norm G_{\op}>3] \le \frac{1}{2} e^{-n}.
    \]
\end{lemma}

Now consider a square, standard complex Gaussian matrix $G\in\C^{k\times k}$, and list its singular values
\[
    0\le\sigma_1(G)\le\cdots\le\sigma_k(G).
\]
For $\eta>0$, define
\[
    b_\eta(G)
    :=\left|\set{j\in[k]:\sigma_j(G) \leq \eta}\right|.
\]
In \cref{sec:determinantal-process} we prove the following bound on the concentration of $b_\eta(G)$.

\begin{lemma}[{Concentration of small singular values, \cref{lem:sing-value-count}}]\label{lem:small-count}
    Fix $0 < u \leq \frac{1}{2\sqrt{2}}$.
    If $G_1,\ldots,G_m$ are independent copies of $G \in \mathbb{C}^{k \times k}$, then
    \[
        \Prob\left[\sum_{i=1}^m b_{u/\sqrt{k}}(G_i)
        \geq \frac m4\right]
        \leq e^{-u^2 m/3}.
    \]
\end{lemma}

\subsection{The Gaussian construction}

The following theorem combines the concentration bounds from last section with the Grassmannian net to argue that a family of $m = \poly(n)$ independent Gaussian matrices $\{Q_1,\dots,Q_m\}$ form a robust subspace design with high probability.

For $Q:\C^n\to\C^k$, $W\in\Gr(k,n)$ and $\eta\in(0,1]$, we will use the notation
\[
    b_\eta(Q,W)
    :=\left|\set{j\in[k]:\sigma_j(\restr QW) \leq \eta}\right|.
\]
Equivalently, if $N\in\C^{n\times k}$ has orthonormal columns spanning $W$, then $b_\eta(Q,W)$ is the number of singular values of $QN$ at most $\eta$.

\begin{theorem}[Gaussian strong robust design]\label{thm:gaussian-strong-robust-design}
    There are absolute constants $C,c>0$ with the following property.
    Fix $1\le k< n$ and
    \[
        m \geq C k (n-k) \log(nk).
    \]
    Choose $Q_1,\ldots,Q_m\in\C^{k\times n}$ independently, with entries
    $\mathcal N_\C(0,1/(9n))$.
    Then, with probability at least
    \[
        1-e^{-cm}-\frac{m}{2}e^{-n},
    \]
    the family $\{Q_1,\ldots,Q_m\}$ is a $(k,\frac{1}{18\sqrt{nk}},m/4)$-strong robust subspace design.
\end{theorem}
\begin{proof}
    Fix constants $\eta=\frac{1}{9\sqrt{nk}}$ and $\rho=\frac{\eta}{2}$.
    Fix $W\in\Gr(k,n)$ and let $N\in\C^{n\times k}$ have orthonormal columns spanning $W$.
    The key fact that we will use is that by unitary invariance of the complex Gaussian distribution,
    \[
        Q_iN = \frac1{3\sqrt n}G_i,
    \]
    for $G_i\in\C^{k\times k}$ with independent standard complex Gaussian entries.
    Since
    \[
        b_{\eta}(Q_i,W)
        = b_{3\eta\sqrt{n}}(G_i)
        = b_{1/(3\sqrt{k})}(G_i),
    \]
    by our choice of $\eta$, we can apply \Cref{lem:small-count} with $u=1/3$ to get
    \[
        \Prob\left[\sum_{i=1}^m b_{\eta}(Q_i,W)\geq\frac m4\right]
        \le e^{-m/27}.
    \]

    Now let $\cN_\rho$ be a $\rho$-net of $\Gr(k,n)$ as in \Cref{lem:net}.  A union
    bound gives
    \begin{align*}
        \Prob\left[ \exists W\in\cN_\rho: \sum_{i=1}^m b_{\eta}(Q_i,W) \geq \frac m4 \right]
        &\leq
        \left(\frac{C_{\mathrm{net}}}{\rho}\right)^{2k(n-k)}e^{-m/27}.
    \end{align*}
    By our choice of $\rho = \Theta(1/\sqrt{nk})$,
    \[
        \log\left(\frac{C_{\mathrm{net}}}{\rho}\right)
        =O(\log(nk)),
    \]
    and so there exist constants $C,c$ such that
    \[
        \left(\frac{C_{\mathrm{net}}}{\rho}\right)^{2k(n-k)}e^{-m/27}
        \leq e^{-c m}
    \]
    provided $m \geq C k (n-k) \log(nk)$.

    By \Cref{lem:gaussian-op} applied to the operator $3 Q_i$ (whose entries are distributed as $\mathcal{N}_{\mathbb C}(0,1/n)$), and a union bound over the $Q_i$'s, it holds that
    \[
        \Prob\left[\exists i \in [m]: \norm{Q_i}_{\op}>1\right]
        \le \frac{m}{2} e^{-n}.
    \]
    Assume that neither failure event occurs.  Let $W\in\Gr(k,n)$ be arbitrary, and
    choose $V\in\cN_\rho$ with
    $d(W,V)\le\rho$.
    For every $i$ and every singular-value index $j$,
    \Cref{lem:lipschitz} gives
    \[
        \abs{\sigma_j(\restr{Q_i}{W})-\sigma_j(\restr{Q_i}{V})}
        \leq \rho=\frac{\eta}{2}.
    \]
    Therefore every singular value of $\restr{Q_i}{W}$ at most $\eta/2$ corresponds to
    a singular value of $\restr{Q_i}{V}$ at most $\eta$, and hence
    \[
        b_{\eta/2}(Q_i,W)
        \leq b_{\eta}(Q_i,V).
    \]
    By our assumption, summing over $i$ yields
    \[
        \sum_{i=1}^m b_{\eta/2}(Q_i,W)
        \leq \sum_{i=1}^m b_{\eta}(Q_i,V)
        \leq \frac m4.
    \]
    This holds simultaneously for every $W$, making the $Q_i$'s a $(k,\eta/2,m/4)$ robust subspace design.
\end{proof}

\subsection{Computability in randomized logspace}

We now show how to obtain a finite-precision version of the Gaussian construction in \cref{thm:gaussian-strong-robust-design} that is computable in randomized logspace.
This will be used in all of our complexity-theoretic applications.

To this end, we first show that a randomized algorithm can output a finite-precision scalar that approximates a continuous Gaussian scalar.
One subtlety in formalizing such an approximation is that the output of a randomized algorithm has a discrete distribution, whereas a Gaussian random variable has a continuous distribution.
We handle this using the notion of a \emph{coupling}.
A coupling of two random variables $X$ and $Y$ is a pair of random variables $(X',Y')$ defined on a common probability space such that $X'$ has the same distribution as $X$ and $Y'$ has the same distribution as $Y$.
As is customary, we will identify $X'$ with $X$ and $Y'$ with $Y$ once such a coupling has been fixed.

\begin{lemma}[Finite-precision Gaussian sampling]
    \label{lem:finite-precision-gaussian-sampling}
    Let $\varepsilon,\delta \in(0,1/2]$.
    There is a randomized algorithm which uses $O(\log(1/\eps) + \log\log(1/\delta))$ workspace, runs in time $\poly(1/\varepsilon,\,\log(1/\delta))$, and outputs a complex rational $\widetilde g\in\mathbb Q(i)$ with the following properties:
    \begin{enumerate}[label=(\roman*)]
        \item it has an $O(\log(1/\eps) + \log\log(1/\delta))$-bit dyadic representation,
        \item its modulus is bounded by $|\widetilde g| = O(\sqrt{\log(1/\delta)})$, and
        \item there is a coupling of \(\widetilde g\) with \(g\sim\mathcal N_\C(0,1)\) such that
        \[
            \mathrm{Pr} \bigl[|\widetilde g-g|>\varepsilon\bigr]\leq \delta.
        \]
    \end{enumerate}
\end{lemma}
\begin{proof}
    We first consider the case of a real Gaussian scalar.
    Let $Z_1,\dots,Z_N$ be independent Bernoulli random variables taking values $0$ and $1$ with equal probability, and set $S_N:=Z_1+\cdots+Z_N$.
    Furthermore, we define
    \[
        \widetilde X_N:=\frac{2S_N-N}{\sqrt N}.
    \]
    By Tusnády's lemma \cite[Proposition T]{MZ12}, there is a coupling of $\widetilde X_N$ with a standard real Gaussian $X\sim\mathcal N_{\mathbb R}(0,1)$ such that
    \[
        |X-\widetilde X_N|
        \leq \frac{2}{\sqrt N} + \frac{X^2}{4\sqrt N}.
    \]

    For $B:=\left\lceil\sqrt{2\log(4/\delta)}\right\rceil$, we find by the standard Gaussian tail bound,
    \[
        \Pr[|X|>B]
        \leq
        2e^{-B^2/2}
        \leq
        \frac{\delta}{2}.
    \]
    We now choose
    \[
        N=2^{2k-1}
        =
        O(B^4/\varepsilon^2)
        =
        O(\log^2(1/\delta)/\varepsilon^2)
    \]
    for some integer $k\geq1$ large enough so that, on the event $|X|\leq B$, $|X-\widetilde X_N|\leq \varepsilon$.
    Consequently,
    \[
        \Pr\left[|X-\widetilde X_N|>\varepsilon\right]
        \leq
        \frac{\delta}{2}.
    \]

    We now pass to a complex Gaussian.
    Take two independent copies $\widetilde X_N$ and $\widetilde Y_N$ of the above construction, coupled respectively with independent standard real Gaussians $X,Y\sim\mathcal N_{\mathbb R}(0,1)$, and define $g:=(X+iY)/\sqrt{2}\sim\mathcal N_{\mathbb C}(0,1)$.
    Furthermore, set
    \[
        \widetilde g := (\widetilde X_N+i\widetilde Y_N)/\sqrt2,
    \]
    whenever both $\widetilde X_N$ and $\widetilde Y_N$ are bounded by $B+\varepsilon$ in modulus, and set $\widetilde g:=0$ otherwise.
    Since $N=2^{2k-1}$, we have
    \[
        \frac{\widetilde X_N}{\sqrt2}
        =
        \frac{2S_N-N}{\sqrt{2N}}
        =
        \frac{2S_N-N}{2^k},
    \]
    and likewise for $\widetilde Y_N/\sqrt2$.
    Hence $\widetilde g$ is a dyadic complex rational. Moreover, since $2S_N-N$ has $O(\log N)$ bits and the denominator is $2^k \leq N$, its dyadic representation uses $O(\log N)=O(\log(1/\varepsilon)+\log\log(1/\delta))$ bits.
    Further,
    \[
        |\widetilde g|
        \leq
        B+\varepsilon
        =
        O\bigl(\sqrt{\log(1/\delta)}\bigr).
    \]
    On the event $|X|\leq B$ and $|Y|\leq B$, we have that $|\widetilde X_N| \leq B+\varepsilon$ and $|\widetilde Y_N|\leq B+\varepsilon$.
    Therefore
    \[
        |g-\widetilde g|^2
        \leq
        \frac{|X-\widetilde X_N|^2}{2}
        +
        \frac{|Y-\widetilde Y_N|^2}{2}
        \leq
        \varepsilon^2.
    \]
    Thus, by a union bound,
    \[
        \Pr[|g-\widetilde g|>\varepsilon]
        \leq \delta.
    \]

    It remains to verify the computational bounds.
    An algorithm can generate the $N$ Bernoulli variables used for the real and imaginary part sequentially and only needs to maintain their running sum and a counter.
    From the final value of the sum, it can directly determine whether the $B+\varepsilon$ truncation condition is satisfied and accordingly output either the resulting dyadic value or $0$.
    It may therefore use $O(\log N) = O(\log(1/\varepsilon) + \log\log(1/\delta))$ workspace and run in time $\poly(1/\varepsilon,\,\log(1/\delta))$.
\end{proof}

The above lemma shows in particular that for inverse-polynomial accuracy and inverse-exponential failure probability, we only require \(O(\log n)\) workspace to approximately sample a Gaussian scalar. Applying the sampler independently entrywise gives the following finite-precision version of the Gaussian construction from \cref{thm:gaussian-strong-robust-design}.

\begin{theorem}[Finite-precision strong robust subspace designs]
    \label{thm:finite-precision-design}
    Let $C,c>0$ be the constants from \cref{thm:gaussian-strong-robust-design}. For every $1 \leq k<n$ and
    \[
        m\geq Ck(n-k)\log(nk),
    \]
    there is a randomized algorithm that outputs $\widetilde Q_1,\ldots,\widetilde Q_m \in\mathbb Q(i)^{k\times n}$ such that, with probability at least
    \[
        1-e^{-cm}-\frac m2e^{-n}-e^{-(m+n)},
    \]
    the family $\{\widetilde Q_1,\ldots,\widetilde Q_m\}$ is a $\left(k,\frac{1}{144\sqrt{nk}},\frac m4 \right)$ strong robust subspace design.

    The algorithm uses $O(\log(nm))$ workspace and runs in $\poly(n,m)$ time. Moreover, all entries of the matrices $\widetilde Q_i$ are complex rationals representable using $O(\log(nm))$ bits.
\end{theorem}
\begin{proof}
    Set
    \[
        \eta:=\frac{1}{18\sqrt{nk}},
        \qquad
        \widetilde\eta:=\frac{\eta}{8}
        =
        \frac{1}{144\sqrt{nk}},
    \]
    and choose
    \[
        d:= \left\lceil 6\sqrt n\right\rceil.
    \]

    Independently for every $\ell\in[m]$, $i\in[k]$, and $j\in[n]$, apply the finite-precision scalar sampler from \cref{lem:finite-precision-gaussian-sampling} to obtain a dyadic complex rational $\widetilde g_{\ell,i,j}$ with accuracy and failure probability
    \[
        \varepsilon := \frac{d}{144nk},
        \qquad
        \delta := \frac{2^{-2(m+n)}}{mkn}.
    \]
    We may couple these variables with independent $g_{\ell,i,j}\sim\cN_\C(0,1)$ such that
    \[
        \Pr\left[
        |\widetilde g_{\ell,i,j}-g_{\ell,i,j}|>\varepsilon
        \right]
        \leq
        \delta
    \]
    for every entry.

    Let $G_\ell,\widetilde G_\ell\in\C^{k\times n}$ denote the matrices with entries $g_{\ell,i,j}$ and $\widetilde g_{\ell,i,j}$, respectively, and define
    \[
        Q_\ell:=\frac{G_\ell}{d},
        \qquad
        \widetilde Q_\ell:=\frac{\widetilde G_\ell}{d}.
    \]
    Let $\mathcal E$ be the event that $|\widetilde g_{\ell,i,j}-g_{\ell,i,j}|\leq\varepsilon$ simultaneously for all $mkn$ entries.
    By a union bound,
    \[
        \Pr[\mathcal E^c]
        \leq
        mkn\cdot\delta
        \leq
        2^{-2(m+n)}
        \leq
        e^{-(m+n)}.
    \]
    On $\mathcal E$, for every $\ell\in[m]$,
    \begin{align*}
        \|\widetilde Q_\ell-Q_\ell\|_{\op}
        \leq \|\widetilde Q_\ell-Q_\ell\|_{\mathrm F}
        \leq \frac{\sqrt{kn}\,\varepsilon}{d}
        \leq \frac{1}{144\sqrt{nk}}
        = \widetilde\eta.
    \end{align*}

    We next show that the $\widetilde Q_\ell$ satisfy the robust-design property.
    By \cref{thm:gaussian-strong-robust-design}, with probability at least $1-e^{-cm}-\frac m2e^{-n}$, the matrices $\frac{G_\ell}{3\sqrt n}$ form a $(k,\eta,m/4)$-strong robust subspace design.
    Since
    \[
        Q_\ell
        =
        \frac{3\sqrt n}{d}\,
        \frac{G_\ell}{3\sqrt n}
    \]
    and $\frac14\leq\frac{3\sqrt n}{d} \leq \frac12$, the matrices $Q_\ell$ form a $\left(k,\frac{\eta}{4},\frac m4\right)$ strong robust subspace design.
    Now fix an arbitrary $k$-dimensional subspace $W\in\Gr(k,n)$.
    By Weyl's inequality, on the event $\mathcal E$ we have, for every $\ell\in[m]$ and $r\in[k]$,
    \[
        \left|
        \sigma_r(\widetilde Q_\ell|_W)
        -
        \sigma_r(Q_\ell|_W)
        \right|
        \leq
        \|\widetilde Q_\ell-Q_\ell\|_{\op}
        \leq
        \widetilde\eta.
    \]
    Consequently, we find that if $\sigma_r(\widetilde Q_\ell|_W)\leq\widetilde\eta$, then also $\sigma_r(Q_\ell|_W)\leq2\widetilde\eta=\frac{\eta}{4}$.
    Hence
    \[
        b_{\widetilde\eta}(\widetilde Q_\ell,W)
        \leq
        b_{\eta/4}(Q_\ell,W).
    \]
    Summing over $\ell$ gives
    \[
        \sum_{\ell=1}^m
        b_{\widetilde\eta}(\widetilde Q_\ell,W)
        \leq
        \sum_{\ell=1}^m
        b_{\eta/4}(Q_\ell,W)
        \leq
        \frac m4.
    \]
    Moreover, in the success event of the Gaussian construction and on event $\mathcal{E}$, $\|\widetilde Q_\ell\|_\op \leq \frac{1}{2}+\widetilde \eta \leq 1$ since $\|Q_\ell\|_\op \leq \frac{1}{2}$.
    Thus, the family $\{\widetilde Q_1,\ldots,\widetilde Q_m\}$ is a $\left(k,\widetilde\eta,\frac m4\right)$ strong robust subspace design, with probability at least
    \[
        1-e^{-cm}-\frac m2e^{-n}-e^{-(m+n)}.
    \]

    It remains to verify the computational bounds.
    To see this, we simply verify that
    \[
        1/\varepsilon
        = O(k\sqrt n)
        = \poly(n),
        \qquad
        \log(1/\delta)
        =
        O(m+n+\log(mkn)).
    \]
    Hence, by \cref{lem:finite-precision-gaussian-sampling}, each entry can be generated in $\poly(n,m)$ time using $O\bigl(\log(1/\varepsilon)+\log\log(1/\delta)\bigr)= O(\log(nm))$ workspace and, has an exact $O(\log(nm))$-bit rational representation.
    The entries can be generated sequentially, using only additional $O(\log(nm))$-bit counters.
\end{proof}

\section{Explicit robust constructions}\label{sec:explicit-construction}

In the following section we describe explicit constructions of robust subspace designs.
First, in \cref{sec:explicit-construction-based-on-coordinate-projections} we describe a very simple construction based on coordinate projections that yields a design of size $\binom{n}{k}$.
Then, in \cref{sec:explicit-construction-using-sketching} we combine this construction with an idea based on sketching to reduce the size to $n^4 k^{O(k)}$.

\subsection{Explicit construction based on coordinate projections}\label{sec:explicit-construction-based-on-coordinate-projections}

In this section, we present a simple construction of a weak robust subspace design.
The construction has a good inverse polynomial robustness parameter $\eta$. However, it is of large size $m=\binom{n}{k}$.
This unfortunately makes it impractical for our later complexity theoretic applications, which require polynomial-sized designs.

For $I\subseteq[n]$ with $|I|=k$, let $P_I:\C^n\to\C^k$ denote the coordinate projection onto $I$. The following theorem is a restatement of \cref{thm:coordinate}.
It uses a lemma about maximum-volume submatrices from \cite{goreinov2010find}, who developed it in the context of matrix approximation theory.

\begin{theorem}
    The family
    \[
        \mathcal P_k=\{P_I: I\subseteq[n],\;|I|=k\}
    \]
    is an~$(k,\eta,\binom{n}{k}-1)$-weak robust subspace design with $\eta=1/(2\sqrt{nk})$.
\end{theorem}
\begin{proof}
    It suffices to show that for any fixed $k$-dimensional subspace $W\subseteq\C^n$, there exists $I \subseteq [n]$ with $|I| = k$ such that $\sigma_{\min}(P_I|_W) \geq 1/\sqrt{nk} > \eta$.
    With $N\in\C^{n\times k}$ having orthonormal columns spanning $W$, we claim that this is satisfied by $I$ maximizing
    \[
        |\!\det(P_I N)|.
    \]
    Following \cite[Lemma 1]{goreinov2010find}, the matrix $P_I N$ is called a \emph{dominant submatrix} of $N$ and its entries satisfy
    \[
        |N (P_I N)^{-1}(i,j)|
        \leq 1.
    \]
    It remains to note that
    \[
        \sigma_{\min}(P_I|_W)
        = \sigma_{\min}(P_I N)
        = 1/\|(P_I N)^{-1}\|_{\op}
    \]
    and that
    \[
        \|(P_I N)^{-1}\|_{\op}
        = \|N (P_I N)^{-1}\|_{\op}
        \leq \|N (P_I N)^{-1}\|_F
        \leq \sqrt{nk}. \qedhere
    \]
\end{proof}

\subsection{Improved explicit construction using sketching} \label{sec:explicit-construction-using-sketching}

We first state the main result of this section.

\begin{theorem}[Explicit robust family]\label{thm:composed-subspace-design}

    For every $1\le k\le n$ there is an explicitly computable family $\cQ_{n,k}\subseteq\mathbb Q^{k\times n}$ of size $m= O(n^4k^{O(k)})$ which is a $(k, \eta, \Delta)$-weak robust subspace design, with parameters
    \begin{equation*}
        \eta = \Omega\left( \frac{1}{k^{3/2}\sqrt n}\right), \qquad \Delta<m.
    \end{equation*}

    Moreover, every entry is rational of polynomial bit length in $n$ and enumerable in space $O(\log n + k\log k)$.
\end{theorem}

The proof uses a composition of two families.
The first family maps into $O(k^2)$ dimensions rather than into $\C^k$.
Effectively, it is a sketching family that robustly embeds the original $n$ dimensions into $O(k^2)$ dimensions.
We then compose this family with our explicit coordinate-projection design to reduce the output dimension to $k$, yielding a valid robust subspace design.

\subsubsection{Robust sketching family}

Our improved explicit construction goes through an intermediate dimension reduction phase based on sketching.
To describe this step, we use the following terminology.
It is effectively a relaxation of a weak robust subspace design where the maps are allowed to have rank $d > k$.

\begin{definition}[Weak robust sketching family] \label{def:sketching-family}
    Let $d\geq k$. A family of linear maps
    \[
        A_1,\ldots,A_m:\C^n\to\C^d
    \]
    is a $(k,\eta,\Delta)$-weak robust sketching family if $\Delta < m$, $\|A_i\|_{\op}\leq 1$ for every $i$ and, for every $k$-dimensional subspace $W\subseteq\C^n$,
    \[
        \bigl|\{i\in[m]:\sigma_{\min}(A_i|_W)\leq\eta\}\bigr|
        \leq \Delta.
    \]
    When $d=k$, this is exactly a $(k,\eta,\Delta)$-weak robust subspace design.
\end{definition}

We next show that an explicit, polynomial-sized family of sketching matrices yields a weak robust sketching family (hence the name).
To establish this, we adapt some ideas from \cite[Section~2.1, Eq.~(1)]{CW09}, and \cite{Sarlos06} in our setting. Our main technical tool is based on the \emph{fourth moment method}.

\begin{definition}[Tug-of-war matrix {\cite{Sarlos06}}]
    A random matrix $S\in\R^{r\times n}$ is called a $t$-tug-of-war matrix if its entries are $t$-wise independent and, for every $a\in[r]$ and $i\in[n]$,
    \[
        \Prob\left(S_{a,i}=\pm\frac1{\sqrt r}\right)=\frac12.
    \]
\end{definition}

Sarl\'os proves the corresponding estimate below over the reals \cite{Sarlos06}. The same fourth-moment calculation gives the following complex version; we include the short argument for completeness.

\begin{lemma}[Expected product of tug-of-war matrices]
    \label{lem:expected-tug-of-war}
    Let $A,B\in\C^{n\times k}$ and let $S\in\R^{r\times n}$ be a $4$-tug-of-war matrix. Then
    \[
        \E\!\left[\|B^*S^*SA-B^*A\|_F^2\right]
        \leq
        \frac{2}{r}\|A\|_F^2\|B\|_F^2.
    \]
\end{lemma}

\begin{proof}
    For $x,y\in\C^n$, write
    \[
        Z:=\langle Sx,Sy\rangle-\langle x,y\rangle.
    \]
    If the $a$-th row of $\sqrt r\,S$ is $\varepsilon_a=(\varepsilon_{a1},\ldots,\varepsilon_{an})$, then
    \[
        Z = \frac 1r
        \sum_{a=1}^r \sum_{i\neq j} \varepsilon_{ai} \varepsilon_{aj} \,\overline{x_i}y_j.
    \]
    By $4$-wise independence, the cross terms corresponding to distinct rows vanish in the second moment. For one row, grouping the terms indexed by $\{i,j\}$ and using $|u+v|^2\leq2(|u|^2+|v|^2)$ gives
    \[
        \E |Z|^2 \leq \frac{2}{r} \|x\|^2\|y\|^2.
    \]
    Apply this estimate to every pair of columns of $A$ and $B$ and sum the resulting bounds over the entries of $B^*S^*SA-B^*A$.
\end{proof}

Here, we will just need to consider the case $A=B=N \in \C^{n \times k}$ for some $N$ satisfying $N^*N = I_k$.
From this bound, we can deduce a lower bound on the minimum singular value with high probability.

\begin{lemma}\label{lem:tow-sigma-min}
    Let $N\in\C^{n\times k}$ have orthonormal columns, let $S\in\R^{r\times n}$ be a $4$-tug-of-war random matrix, and let $0\leq\beta<1$. Then
    \[
        \Prob\!\left[\sigma_{\min}(SN)\leq\sqrt\beta\right]
        \leq \frac{2k^2}{(1-\beta)^2r},
    \]
    and $\|S\|_{\op}\leq\sqrt n.$
\end{lemma}
\begin{proof}
    The second bound follows from $\|S\|_{\op}^2 \leq\|S\|_F^2 = rn\cdot\frac1r = n$.

    For the first bound, set
    \[
        G:=N^*S^*SN=(SN)^*(SN)\in\C^{k\times k}.
    \]
    By \cref{lem:expected-tug-of-war} with $A=B=N$,
    \[
        \E\|G-I_k\|_F^2
        \leq
        \frac{2\|N\|_F^4}{r}
        =
        \frac{2k^2}{r}.
    \]
    Hence, for every $a>0$, Markov's inequality gives
    \begin{equation}\label{eq:frobenius-markov-bound}
        \Prob\!\left[\|G-I_k\|_{\op}^2\geq a\right]
        \leq
        \Prob\!\left[\|G-I_k\|_F^2\geq a\right]
        \leq
        \frac{2k^2}{ar}.
    \end{equation}
    On the other hand, for every $x\in\C^k$, $x^*Gx \geq \bigl(1-\|G-I_k\|_{\op}\bigr)\|x\|^2$, and therefore
    \[
        \sigma_{\min}(SN)^2
        =
        \lambda_{\min}(G)
        \geq
        1-\|G-I_k\|_{\op}.
    \]
    Consequently, by Equation \eqref{eq:frobenius-markov-bound},
    \[
        \begin{aligned}
            \Prob\!\left[\sigma_{\min}(SN)^2\leq\beta\right]
            \leq
            \Prob\!\left[\|G-I_k\|_{\op}\geq1-\beta\right]
            =
            \Prob\!\left[\|G-I_k\|_{\op}^2\geq(1-\beta)^2\right]
            \leq
            \frac{2k^2}{(1-\beta)^2r}.
        \end{aligned} \qedhere
    \]
\end{proof}

Having established the robust embedding properties of a 4-tug-of-war matrix ensemble, it remains to notice that there exists a logspace-enumerable and polynomial-sized instantiation of this ensemble.
It follows from a standard construction of Alon--Babai--Itai
\cite[Proposition~6.3]{ABI86} giving a logspace-enumerable uniform sample space ${\cal T} = \{t \in \{\pm1\}^M\}$ of size $|{\cal T}| \leq 16M^4$ with vectors whose coordinates are 4-wise independent unbiased signs. More precisely, let $\ell=\lceil \log_2 M\rceil$ and $q=2^\ell$.
Taking $d=4$ and this choice of $q$ in
\cite[Proposition~6.3]{ABI86} yields $4$-wise independent random
variables uniformly distributed over $\mathbb{F}_q$ (equivalently, over
$[q]$ after fixing an identification). To obtain unbiased signs, we apply
the standard field trace
\[
    \operatorname{Tr}_{\mathbb{F}_q/\mathbb{F}_2}:
    \mathbb{F}_q\to\mathbb{F}_2,
    \qquad
    \operatorname{Tr}(z)
    =
    z+z^2+\cdots+z^{2^{\ell-1}},
\]
and map $b\in\mathbb{F}_2$ to $(-1)^b$; see, e.g., \cite{DPT24}.
Since the trace is a surjective $\mathbb{F}_2$-linear map, this preserves
$4$-wise independence and gives unbiased signs.

At this point, we have all the ingredients to construct and analyze the sketching family.

\begin{theorem}[Robust sketching family]\label{thm:tow-subspace}
    For every $1\leq k\leq n$, let $r:=16k^2$.
    There is an indexed family
    \[
        S_1,\ldots,S_m \in \mathbb Q^{r\times n},
        \qquad
        m\leq16(rn)^4,
    \]
    which is a $(k,\eta,m/2)$-weak robust sketching family, with $\eta=\Omega(1/\sqrt n)$.
    The family is enumerable in logspace and its coefficients can be computed in logspace.
\end{theorem}

\begin{proof}
    By \cite{ABI86}, we can generate in logspace an indexed family of sign matrices
    \[
        E_1,\ldots,E_m\in\{\pm1\}^{r\times n},
        \qquad
        m=O((rn)^4),
    \]
    such that the $rn$ coordinates of a uniformly chosen $E_i$ are $4$-wise independent unbiased signs.
    Recalling that $r=16k^2$, and we can then define $S_i^{(0)}:=\frac1{4k}E_i$ so that a uniformly chosen $S_i^{(0)}$ is a $4$-tug-of-war matrix.

    Let $d$ be the smallest power of two satisfying
    $\sqrt n\leq d<2\sqrt n$.
    Define $S_i:=\frac1dS_i^{(0)}$.
    Every entry of $S_i$ is rational of polynomial bit length, and
    \[
        \|S_i\|_{\op}
        \leq
        \frac{\sqrt n}{d}
        \leq1.
    \]

    Now fix a $k$-dimensional subspace $W\subseteq\C^n$ and let $N\in\C^{n\times k}$ have orthonormal columns spanning $W$.
    Applying \cref{lem:tow-sigma-min} with $\beta=1/2$ and $r=16k^2$,
    \[
        \Prob\!\left[
        \sigma_{\min}(S^{(0)}N)\leq\frac1{\sqrt2}
        \right]
        \leq
        \frac{2k^2}{(1/2)^2\,16k^2}
        =
        \frac12.
    \]
    Therefore at most half of the indices satisfy
    \[
        \sigma_{\min}(S_iN)
        \leq
        \frac1{\sqrt2\,d}.
    \]
    Since $d<2\sqrt n$, we may take $\eta:=\frac1{\sqrt2\,d} = \Omega(1/\sqrt n)$, establishing that the family is a $(k,\eta,m/2)$-weak robust sketching family.
\end{proof}

\subsection{Composing families}

For the rank-completion application, the final maps must have output dimension $k$ (and hence rank at most $k$), whereas the sketching family has output dimension $r=\Theta(k^2)$. We now combine the sketching family with the coordinate construction in order to reduce the output dimension from $r$ to $k$.

\begin{theorem}[Composition]\label{thm:composition-subspace-designs}
    Let $A_1,\ldots,A_{m_A}\in\C^{n_{AB}\times n}$ be an $(m_A)$-sized $(k,\eta_A,\Delta_A)$-weak robust sketching family, and let $B_1,\ldots,B_{m_B}\in\C^{k\times n_{AB}}$ be an $(m_B)$-sized $(k,\eta_B,\Delta_B)$-weak robust subspace design, with $\Delta_A<m_A$ and $\Delta_B<m_B$.

    Then
    \[
        \{B_jA_i:i\in[m_A],\,j\in[m_B]\}
    \]
    is a $(k,\eta_{AB},\Delta_{AB})$-weak robust subspace design, where $m_{AB}=m_Am_B$ and $\eta_{AB}=\eta_A\eta_B$.
    Moreover,
    \[
        \Delta_{AB}
        =
        m_A\Delta_B+m_B\Delta_A-\Delta_A\Delta_B
        <
        m_{AB}.
    \]
    If the two families can be enumerated and their entries computed in logspace, then the composed family has the same property.
\end{theorem}
\begin{proof}
    The size follows directly, and the operator norms are bounded by $1$ because the operator norm is submultiplicative.

    Fix a $k$-dimensional subspace $W\subseteq\C^n$. There are at least $m_A-\Delta_A$ indices $i$ such that
    \[
        \sigma_{\min}(A_i|_W) > \eta_A.
    \]
    Fix such an $i$. Then $A_i|_W$ is injective, so $A_iW$ is a $k$-dimensional subspace of $\C^{n_{AB}}$. By the design property of the $B_j$'s, at least $m_B-\Delta_B$ indices $j$ satisfy
    \[
        \sigma_{\min}(B_j|_{A_iW}) > \eta_B.
    \]
    For every such pair $(i,j)$ and every non-zero $x\in W$,
    \[
        \begin{aligned}
            \|B_jA_ix\|
            &\geq
            \sigma_{\min}(B_j|_{A_iW})\,\|A_ix\|
            \geq
            \sigma_{\min}(B_j|_{A_iW})
            \sigma_{\min}(A_i|_W)\,\|x\|
            > \eta_B\eta_A\,\|x\|.
        \end{aligned}
    \]
    Hence there are at least
    \[
        (m_A-\Delta_A)(m_B-\Delta_B)
    \]
    good pairs. Equivalently, the number of bad pairs is at most
    \[
        \begin{aligned}
            m_Am_B-(m_A-\Delta_A)(m_B-\Delta_B)
            &=
            m_A\Delta_B+m_B\Delta_A-\Delta_A\Delta_B.
        \end{aligned}
    \]
    This proves the claimed parameters.

    For the final assertion, the pairs $(i,j)$ can be enumerated in logspace, and each entry of $B_jA_i$ can be computed entrywise in logspace from the corresponding entries of $A_i$ and $B_j$.
\end{proof}

We can now prove our main theorem.

\begin{proof}[Proof of \cref{thm:composed-subspace-design}]
    We compose the two constructions above.
    First, \cref{thm:tow-subspace} gives a robust sketching family $A_i:\C^n\to\C^{n_{AB}}$ with
    \[
        n_{AB}=r=16k^2,
        \qquad
        m_A=O((rn)^4),
        \qquad
        \eta_A=\Omega(1/\sqrt n),
        \qquad
        \Delta_A=m_A/2.
    \]
    Second, \cref{thm:coordinate}, applied in ambient dimension $n_{AB}$, gives a coordinate subspace design $B_j:\C^{n_{AB}}\to\C^k$ with
    \[
        m_B=\binom{n_{AB}}k,
        \qquad
        \eta_B=\frac1{2\sqrt{n_{AB}k}},
        \qquad
        \Delta_B=m_B-1.
    \]
    By \cref{thm:composition-subspace-designs}, their compositions form a weak robust subspace design with
    \[
        m=m_Am_B
        =
        O\!\left(
        (k^2n)^4\binom{16k^2}{k}
        \right)
        =
        O\!\left(n^4k^{O(k)}\right),
    \]
    where we used $\binom{16k^2}{k}\leq(16ek)^k$.
    Moreover,
    \[
        \eta
        =
        \eta_A\eta_B
        =
        \Omega\!\left(
        \frac1{\sqrt n}\cdot\frac1{\sqrt{16k^3}}
        \right)
        =
        \Omega\!\left(\frac1{k^{3/2}\sqrt n}\right).
    \]
    Finally, $\Delta = m_A\Delta_B+m_B\Delta_A-\Delta_A\Delta_B < m_Am_B=m$.

    All matrices have rational entries of polynomial bit length, and the construction is enumerable in space $O(\log m) = O(\log n + k \log k)$ by the corresponding properties of the two individual families and the composition theorem.
\end{proof}

\section{Isolating a unique small quantum witness and other applications} \label{sec:isolating-a-unique-small-quantum-witness-and-other-applications}

In this section we show our main application of robust subspace designs and kernel peeling, to quantum witness isolation, as well as new linear algebra containment results.
Before getting into this, let us recall some basic complexity theory notation:

We use $\mathsf L$, $\mathsf{RL}$, $\mathsf{coRL}$, and $\mathsf{BPL}$ for deterministic logspace and its randomized one-sided- and two-sided-error variants.
These machines have read-only access to the $n$-bit input, and an additional work tape where they may only use $O(\log n)$ tape cells.
Randomized machines run in polynomial time, and error probabilities are at most $1/3$. A family of circuits or matrices is \emph{logspace uniform} if a deterministic logspace algorithm can output its description on a write-only uni-directional tape.

In an oracle class such as $\mathsf{RL}^{\mathcal O}$, the machine has access to a fixed problem in $\mathcal O$ and may write a polynomial-length query on a write-only oracle tape that is not counted as workspace. If the oracle problem has a promise, answers to queries outside that promise may be arbitrary. All claimed correctness and error guarantees must hold regardless of those answers.

\subsection{Quantum space-bounded complexity classes}

Let $s(n)=\Omega(\log n)$ be a space bound, and let $c,k:\mathbb{N}\rightarrow[0,1]$ satisfy $c(n)>k(n)$ for every $n$.
The functions $c$ and $k$ are the completeness and soundness thresholds, respectively. We assume that $s$, $c$, and $k$ are computable in space $O(s(n))$. All unitary circuits use a fixed finite universal gate set with entries in $\mathbb Q(i)$.

Following \cite{FeffermanRemscrim2021}, we use space-uniform families of unitary quantum circuits. Specifically, a family $\{Q_x\}$ is space-$s$ uniform if a deterministic Turing machine, on input $x\in\{0,1\}^n$, runs in space $O(s(n))$ and outputs the description of $Q_x$. Each circuit acts on $O(s(n))$ qubits and consists of at most $2^{O(s(n))}$ gates.

\begin{definition}[Quantum space-bounded computation]
    A promise problem $L=(L_{\mathrm{yes}},L_{\mathrm{no}})$ belongs to $\mathsf{Q_U SPACE}_{c,k}(s(n))$ if there exists a space-$s$ uniform family of unitary quantum circuits $\{Q_x\}$ such that
    \begin{itemize}
        \item if $x\in L_{\mathrm{yes}}$, then $\Pr[Q_x\text{ accepts}]\geq c(n)$,
        \item if $x\in L_{\mathrm{no}}$, then $\Pr[Q_x\text{ accepts}]\leq k(n)$.
    \end{itemize}
    Here $Q_x$ acts on $a_n=O(s(n))$ qubits initialized in $\ket{0}^{\otimes a_n}$ and accepts if a final computational-basis measurement of the first qubit yields outcome $1$.
\end{definition}

For the witness-based classes below, the circuit $Q_x$ acts on $a_n$ auxiliary qubits and $w_n$ witness qubits, where $a_n+w_n=O(s(n))$. Given a witness $\ket{\psi}\in\mathbb{C}^{2^{w_n}}$, the circuit is initialized to $\ket{0}^{\otimes a_n}\ket{\psi}$, and it accepts if the first of the auxiliary qubits is measured to be in state $\ket{1}$ at the end of the computation.
Throughout, all witness states are understood to be normalized.

\begin{definition}[Quantum space-bounded Merlin--Arthur protocols]
    A promise problem $L=(L_{\mathrm{yes}},L_{\mathrm{no}})$ belongs to $\mathsf{Q_U MASPACE}_{c,k}(s(n))$ if there exists a space-$s$ uniform family of unitary verification circuits $\{Q_x\}$ such that
    \begin{itemize}
        \item if $x\in L_{\mathrm{yes}}$, then there exists a witness $\ket{\psi}\in\mathbb{C}^{2^{w_n}}$ with $\Pr[Q_x\text{ accepts }\ket{\psi}]\geq c(n)$,
        \item if $x\in L_{\mathrm{no}}$, then every witness $\ket{\psi}\in\mathbb{C}^{2^{w_n}}$ satisfies $\Pr[Q_x\text{ accepts }\ket{\psi}]\leq k(n)$.
    \end{itemize}
\end{definition}

We next introduce space-bounded analogues of the classes $\mathsf{FewQMA}$ and $\mathsf{UQMA}$ introduced in \cite{jain2011poweruniquequantumwitness}.
These classes mirror the classical restrictions underlying $\mathsf{FewP}$ (introduced in \cite{All86}) and $\mathsf{UP}$ (introduced in \cite{VALIANT197620}), and their logspace counterparts $\mathsf{FewL}$ and $\mathsf{UL}$ (introduced in \cite{BJLR91}): the number of accepting computation paths is replaced by the dimension of an accepting witness subspace $W_x$, which is then restricted to either polynomial dimension ($\mathsf{FewQMA}$) or dimension 1 ($\mathsf{UQMA}$).

In the quantum setting, we additionally require an acceptance gap between the accepting subspace and its orthogonal complement.
To distinguish this requirement from soundness on no-instances, let $\ell:\mathbb{N}\rightarrow[0,1]$ be computable in space $O(s(n))$, with $\ell(n)<c(n)$.
This parameter bounds the acceptance probability of witnesses orthogonal to the accepting subspace on yes-instances.

\begin{definition}[Few quantum witnesses]
    A promise problem $L=(L_{\mathrm{yes}},L_{\mathrm{no}})$ belongs to $\mathsf{FewQ_U MASPACE}_{c,\ell,k}(s(n))$ if there exist a polynomial $q$ and a space-$s$ uniform family of unitary verification circuits $\{Q_x\}$ such that
    \begin{itemize}
        \item if $x\in L_{\mathrm{yes}}$, then there exists a subspace $W_x\subseteq\mathbb{C}^{2^{w_n}}$ with $1\leq\dim W_x\leq q(n)$ such that
        \[
            \Pr[Q_x\text{ accepts }\ket{\psi}]\geq c(n)
            \quad\text{for every }\ket{\psi}\in W_x,
        \]
        and
        \[
            \Pr[Q_x\text{ accepts }\ket{\phi}]\leq \ell(n)
            \quad\text{for every }\ket{\phi}\in W_x^\perp,
        \]
        \item if $x\in L_{\mathrm{no}}$, then every witness $\ket{\psi}\in\mathbb{C}^{2^{w_n}}$ satisfies $\Pr[Q_x\text{ accepts }\ket{\psi}]\leq k(n)$.
    \end{itemize}
\end{definition}

\begin{definition}[A unique quantum witness]
    A promise problem $L=(L_{\mathrm{yes}},L_{\mathrm{no}})$ belongs to $\mathsf{UQ_U MASPACE}_{c,\ell,k}(s(n))$ if there exists a space-$s$ uniform family of unitary verification circuits $\{Q_x\}$ such that
    \begin{itemize}
        \item if $x\in L_{\mathrm{yes}}$, then there exists a one-dimensional subspace $W_x\subseteq\mathbb{C}^{2^{w_n}}$ such that
        \[
            \Pr[Q_x\text{ accepts }\ket{\psi}]\geq c(n)
            \quad\text{for every }\ket{\psi}\in W_x,
        \]
        and
        \[
            \Pr[Q_x\text{ accepts }\ket{\phi}]\leq \ell(n)
            \quad\text{for every }\ket{\phi}\in W_x^\perp,
        \]
        \item if $x\in L_{\mathrm{no}}$, then every witness $\ket{\psi}\in\mathbb{C}^{2^{w_n}}$ satisfies $\Pr[Q_x\text{ accepts }\ket{\psi}]\leq k(n)$.
    \end{itemize}
\end{definition}

We will focus on the special cases:
\[
    \begin{aligned}
        \mathsf{BQ_U L}
        &:=\mathsf{Q_U SPACE}_{2/3,\,1/3}(\log n),\\
        \mathsf{Q_U MAL}
        &:=\mathsf{Q_U MASPACE}_{2/3,\,1/3}(\log n),\\
        \mathsf{Q_U MAL}_1
        &:=\mathsf{Q_U MASPACE}_{1,\,1/3}(\log n),\\
        \mathsf{FewQ_U MAL}_1
        &:=\mathsf{FewQ_U MASPACE}_{1,\,1/3,\,1/3}(\log n),\\
        \mathsf{UQ_U MAL}_1
        &:=\mathsf{UQ_U MASPACE}_{1,\,1/3,\,1/3}(\log n).
    \end{aligned}
\]
Here we note that, for logarithmic space, the polynomial upper bound on $\dim W_x$ is automatic, since the witness register has polynomial dimension.
The only additional restriction in $\mathsf{FewQ_U MAL_1}$ compared to $\mathsf{Q_U MAL_1}$ is therefore the acceptance
bound on $W_x^\perp$ on yes-instances.

Allowing general quantum channels, including intermediate measurements, gives the corresponding classes without the unitary subscript $\mathsf{U}$.
Fefferman and Remscrim \cite[Lemma~18]{FeffermanRemscrim2021} showed that
\[
    \mathsf{BQL}
    =\mathsf{QMAL}
    =\mathsf{BQ_U L}
    =\mathsf{Q_U MAL}.
\]
Consequently, the bounded-error versions of our Few and Unique
classes also coincide with $\mathsf{BQ_U L}$.
The key idea in \cite{FeffermanRemscrim2021} behind removing a logarithmic-size quantum witness is to first amplify the verifier using space-efficient error reduction (see e.g. \cite{fefferman_et_al:LIPIcs.ICALP.2016.14}) and then replace the witness by the maximally mixed state. Since the witness space has polynomial dimension, this leaves an inverse-polynomial acceptance gap, which can again be amplified to obtain a bounded-error algorithm.
This procedure does not, however, preserve perfect completeness. Our focus is therefore on the perfect-completeness witness
classes, for which we have
\[
    \mathsf{UQ_U MAL}_1
    \subseteq \mathsf{FewQ_U MAL}_1
    \subseteq \mathsf{Q_U MAL}_1
    \subseteq \mathsf{Q_U MAL}.
\]
Note that the bounded-error equivalence between space-bounded unitary circuits and general quantum channels does not establish $\mathsf{Q_U MAL}_1=\mathsf{QMAL}_1$.

One motivation for studying $\mathsf{Q_U MAL}_1$ is that it completely characterizes the complexity of testing the singularity of well-conditioned matrices \cite{FeffermanRemscrim2021}.
This makes it a natural quantum counterpart to $\mathsf{C_{=}L}$, the classical counting class characterized by testing whether a general (potentially ill-conditioned) matrix is singular (see e.g. \cite{allender_complexity_1999}).
We show that $\mathsf{FewQ_U MAL}_1$ and $\mathsf{UQ_U MAL}_1$ likewise admit complete problems obtained by imposing additional spectral-gap and kernel-dimension promises.
We will use this characterization to prove our main application, \cref{thm:fewqumal1-in-rl-to-uqumal1}.
To state these characterizations, we introduce the following promise problems:

\medskip
\noindent
\begin{tabular}{@{}p{0.33\linewidth}@{\hspace{1.0em}}l@{\hspace{1.2em}}p{0.53\linewidth}@{}}
    \addlinespace[0.5em]
    \hline
    \addlinespace[0.5em]
    $\kappa$-$\mathsf{conditioned}$-$\mathsf{SINGULAR}$
    & \textit{Input:}
    & $A\in\mathbb{C}^{n\times n}$ with $\|A\|_{\op} \leq 1$. \\

    & \textit{Yes:}
    & $\dim\ker(A)\geq 1$. \\

    & \textit{No:}
    & $\dim\ker(A)=0$ and $\smin(A) \geq 1/\kappa(n)$. \\
    \addlinespace[0.5em]
    \hline
    \addlinespace[0.5em]

    $\kappa$-$\mathsf{conditioned}$-$\mathsf{FewSINGULAR}$
    & \textit{Input:}
    & $A\in\mathbb{C}^{n\times n}$ with $\|A\|_{\op} \leq 1$ and $\splus(A) \geq 1/\kappa(n)$. \\

    & \textit{Yes:}
    & $\dim\ker(A)\geq 1$. \\

    & \textit{No:}
    & $\dim\ker(A)=0$. \\
    \addlinespace[0.5em]
    \hline
    \addlinespace[0.5em]

    $\kappa$-$\mathsf{conditioned}$-$\mathsf{UniqueSINGULAR}$
    & \textit{Input:}
    & $A\in\mathbb{C}^{n\times n}$ with $\|A\|_{\op} \leq 1$ and $\splus(A) \geq 1/\kappa(n)$. \\

    & \textit{Yes:}
    & $\dim\ker(A)= 1$. \\

    & \textit{No:}
    & $\dim\ker(A)=0$. \\
    \addlinespace[0.5em]
    \hline
    \addlinespace[0.5em]
\end{tabular}
\medskip

We write $\pcSINGULAR, \pcFewSINGULAR$ and $\pcUSINGULAR$ for the above problems with polynomial conditioning parameter, i.e. for the case $\kappa(n) = \poly(n)$.

\begin{restatable}[Completeness of poly-conditioned singularity]{lemma}{conditionedsingularitycompleteness}
    \label{thm:conditioned-singularity-completeness}
    Under deterministic logspace many-one reductions,
    \begin{enumerate}[label=(\roman*)]
        \item $\pcSINGULAR$ is $\mathsf{Q_U MAL}_1$-complete (proven in \cite[Lemma~35]{FeffermanRemscrim2021}),
        \item $\pcFewSINGULAR$ is $\mathsf{FewQ_U MAL}_1$-complete, and
        \item $\pcUSINGULAR$ is $\mathsf{UQ_U MAL}_1$-complete.
    \end{enumerate}
\end{restatable}
The hardness proofs for parts~$(ii)$ and~$(iii)$ refine the circuit-to-Hamiltonian argument used for part~$(i)$ in \cite[Lemma~35]{FeffermanRemscrim2021}.
The additional observations are that the construction preserves the dimension of the perfectly accepting witness subspace and that an acceptance gap on its orthogonal complement yields a spectral gap also on yes-instances.
The corresponding containment results follow from the kernel-verification procedure also invoked in \cite[Lemma~35]{FeffermanRemscrim2021}.
We defer the proofs to \cref{sec:completeness-of-conditioned-singularity}.

\subsection{Isolating a unique small quantum witness}

We can now finally prove our main application of robust subspace designs from the introduction:
\qwitnessisolation*
\begin{proof}
    We prove that $\pcFewSINGULAR$ can be decided by a randomized logspace machine with access to an oracle in $\mathsf{UQ_U MAL}_1$.
    The result then follows since $\pcFewSINGULAR$ is $\mathsf{FewQ_U MAL_1}$-complete by \cref{thm:conditioned-singularity-completeness}.

    \textbf{The oracle problem:}
    We introduce an auxiliary promise problem $\pcUSINGULAR'$.
    Its input consists of a positive semidefinite matrix $H\in\mathbb C^{n\times n}$ and a matrix $Q\in\mathbb C^{k'\times n}$, where $0\leq k'\leq n$.
    Writing $M=H+Q^*Q$, the input promise is that for all entries $i,j$
    \[
        |Q(i,j)| \leq 1,
        \qquad
        \|H\|_{\op} \leq 1,
        \qquad
        \splus(M)\geq1/\poly(n).
    \]
    The yes-instances satisfy $\dim\ker(M)=1$, and the no-instances satisfy $\dim\ker(M)=0$.

    This problem belongs to $\mathsf{UQ_U MAL}_1$.
    Indeed, given $(H,Q)$ on the input tape, a deterministic logspace machine can compute the entries of
    \[
        \widehat M:=\frac{H+Q^*Q}{1+nk'}.
    \]
    Then $\|\widehat M\|_{\op}\leq1$ since $\|Q^*Q\|_\op \leq nk'$, $\ker(\widehat M)=\ker(M)$, and $\splus(\widehat M)\geq1/\poly(n)$, yielding an immediate reduction to $\pcUSINGULAR$, which belongs to $\mathsf{UQ_U MAL}_1$ by \cref{thm:conditioned-singularity-completeness}.
    The purpose of supplying~$H$ and~$Q$ separately is that the oracle can reread the entries of~$Q$ to obtain easy access to $Q^*Q$, whereas the randomized reduction only needs to stream the entries of $Q$ once.

    \textbf{The randomized reduction:}
    Let $A\in\mathbb C^{n\times n}$ be a promised input to $\pcFewSINGULAR$, and put
    \[
        H=A^*A,
        \qquad
        k=\dim\ker(A).
    \]
    Then
    \[
        H\succeq0,\qquad
        \|H\|_{\op}\leq1,\qquad
        \ker(H)=\ker(A),\qquad
        \lambda_{\min}^+(H)=\splus(A)^2\geq1/\poly(n).
    \]

    For each $k'\in\{1,\ldots,n-1\}$, we use \cref{thm:finite-precision-design} to generate a strong robust subspace design in randomized logspace of the form
    \[
        \mathcal Q^{(k')}
        =\{Q_1^{(k')},\ldots,Q_{m_{k'}}^{(k')}\}
        \subseteq\mathbb C^{k'\times n},
    \]
    where $m_{k'}=\poly(n)$ and the robustness parameter $\eta_{k'}$ is inverse-polynomial.
    The failure bound in \cref{thm:finite-precision-design} ensures this happens with probability at least 2/3 for all sufficiently large $n$.
    In particular every entry of each $Q_i^{(k')}$ is of modulus at most $1$ in that case. Smaller dimensions can be handled by padding $A$ with an identity block.
    For $k'=0$, take the family consisting of just the zero map.

    For each generated matrix $Q_i^{(k')}$, we write $(H,Q_i^{(k')})$ to the oracle tape, except we replace entries of modulus greater than $1$ by $0$.
    We then query $\pcUSINGULAR'$ with $(H,Q_i^{(k')})$ as input.
    We accept if any query returns yes, and reject otherwise.
    It is easy to see that the reduction uses logarithmic workspace and polynomial time:
    The entries of $H=A^*A$ can be computed in deterministic logspace from the input, and each $Q_i^{(k')}$ is streamed directly to the oracle tape as it is generated.
    Beyond the workspace needed for these procedures, only loop indices need to be stored. Hence, the total space usage is only $O(\log n)$.

    \textbf{Soundness:}
    Suppose that $A$ is a no-instance.
    Then $H$ is positive definite and
    \[
        \lambda_{\min}(H)=\smin(A)^2\geq1/\poly(n).
    \]
    For every generated matrix $Q$, regardless of whether the robust-design construction succeeds, $H+Q^*Q\succeq H$.
    Hence $H+Q^*Q$ is invertible and
    \[
        \lambda_{\min}(H+Q^*Q)
        \geq\lambda_{\min}(H)
        \geq1/\poly(n).
    \]
    Every query is therefore a promised no-instance of $\pcUSINGULAR'$, and the algorithm rejects with probability~$1$.

    \textbf{Completeness:}
    Suppose that $A$ is a yes-instance, so $k\geq1$.
    If $k=1$, the query with $Q=0$ is a promised yes-instance, and the algorithm accepts.

    If $k\geq2$, consider the family generated for $k'=k-1$.
    With probability at least $2/3$, this family has the required robust-design property.
    On this event, the kernel peeling lemma, \cref{lem:kernel-peeling}, specifically its Hermitian version \cref{rem:psd-kernel-peeling}, gives some $Q_i^{(k-1)}$ such that
    \[
        \dim\ker\bigl(H+(Q_i^{(k-1)})^*Q_i^{(k-1)}\bigr)=1
    \]
    and
    \[
        \lambda_{\min}^+
        \bigl(H+(Q_i^{(k-1)})^*Q_i^{(k-1)}\bigr)
        \geq1/\poly(n).
    \]
    This query is a promised yes-instance of $\pcUSINGULAR'$, so the algorithm accepts.
    Other queries may lie outside the oracle promise, but their answers cannot prevent acceptance.
    Thus the algorithm accepts with probability at least $2/3$.
\end{proof}

Since the only place where we use randomness in the reduction above is the construction of the robust subspace designs, we immediately find the following:

\begin{corollary} \label{cor:robust-l}
    If there exists a logspace-uniform weak robust subspace design of polynomial size with inverse polynomial robustness parameter for every ambient subspace dimension $1\leq k< n$, then $\mathsf{FewQ_U MAL_1} \subseteq \mathsf{L}^{\mathsf{UQ_U MAL_1}}$.
\end{corollary}

Instead of computing the robust subspace designs used in \cref{thm:fewqumal1-in-rl-to-uqumal1}, we can also supply it as an advice string to obtain:
\begin{corollary}
    $\mathsf{FewQ_U MAL_1} \subseteq \mathsf{L}^{\mathsf{UQ_U MAL_1}}/\poly$.
\end{corollary}

Here we use standard non-uniform logspace notation.
A promise problem $L=(L_{\mathrm{yes}}, L_{\mathrm{no}})$ is in $\L/\poly$ if there exist advice strings $a(n)\in\{0,1\}^{\poly(n)}$, and a deterministic logspace machine $M$ such that, for every $x\in L_{\mathrm{yes}} \cup L_{\mathrm{no}}$, $x\in L$ if and only if $M(x,a(|x|))$ accepts.
Similarly, $\L^{\mathsf{UQ_{U}MAL}_1}/\poly$, or $\L^{\mathsf{Q_{U}MAL}_1}/\poly$ denote deterministic logspace with polynomial-size advice and oracle access to a promise problem in $\mathsf{UQ_{U}MAL}_1$ or $\mathsf{Q_{U}MAL}_1$, respectively.

\subsection{Other applications}
In this subsection we show that robust and non-robust subspace designs and our kernel peeling lemma help us to solve linear algebra problems such as nullity testing and linear system feasibility in logspace.
The non-robust subspace designs based on the Wronskian construction from \cref{thm:wronskian-subspace-design} allow us to put nullity testing in $\mathsf{C_= L}$ and linear system feasibility in $\mathsf{L^{C_= L}}$, recovering the bounds from \cite{allender_complexity_1999}.
Robust subspace design allow us to put well-conditioned versions of nullity testing and feasibility in $\mathsf{coRL^{Q_U MAL_1}}$ and $\mathsf{BPL^{Q_U MAL_1}}$, respectively.
Before getting into the proofs, let us first formalize the problem statements:

\medskip
\noindent
\begin{tabular}{@{}p{0.5\linewidth}@{\hspace{-10em}}l@{\hspace{0.5em}}p{\dimexpr0.59\linewidth+2.8em\relax}@{}}
    \addlinespace[0.5em]
    \hline
    \addlinespace[0.5em]
    $\mathsf{NULL}_k$
    & \textit{Input:}
    & $A\in\mathbb{C}^{n\times n}$. \\

    & \textit{Yes:}
    & $\dim\ker(A)\geq k$. \\

    & \textit{No:}
    & $\dim\ker(A)<k$. \\
    \addlinespace[0.5em]
    \hline
    \addlinespace[0.5em]

    $\mathsf{FSLE}$
    & \textit{Input:}
    & $A\in\mathbb{C}^{n\times n}$ and $b\in\mathbb{C}^{n}$. \\
    & \textit{Yes:}
    & $\exists x\in\mathbb{C}^{n}: Ax=b$. \\
    & \textit{No:}
    & $\forall x\in\mathbb{C}^{n}: Ax\neq b$. \\
    \addlinespace[0.5em]
    \hline
    \addlinespace[0.5em]

    $\kappa$-$\mathsf{conditioned}$-$\mathsf{NULL}_k$
    & \textit{Input:}
    & $A\in\mathbb{C}^{n\times n}$ with $\|A\|_{\op} \leq 1$. \\

    & \textit{Yes:}
    & $\dim\ker(A)\geq k$. \\

    & \textit{No:}
    & $\dim\ker(A)<k$ and $\sigma_k(A)\geq 1/\kappa(n)$. \\
    \addlinespace[0.5em]
    \hline
    \addlinespace[0.5em]

    $(\kappa,\varepsilon)$-$\mathsf{conditioned}$-$\mathsf{FSLE}$
    & \textit{Input:}
    & $A\in\mathbb{C}^{n\times n}$ and $b\in\mathbb{C}^{n}$ with $\|A\|_{\op},\|b\|\leq 1$ and $\splus(A) \geq 1/\kappa(n)$. \\

    & \textit{Yes:}
    & $\exists x\in\mathbb{C}^{n}: Ax=b$. \\

    & \textit{No:}
    & $\forall x\in\mathbb{C}^{n}: \|Ax-b\|\geq\varepsilon(n)$. \\

    \addlinespace[0.5em]
    \hline
    \addlinespace[0.5em]
\end{tabular}
\medskip

We write $\mathsf{pc}$-$\mathsf{NULL}_k$ and $\mathsf{pc}$-$\mathsf{FSLE}$ for $\kappa$-$\mathsf{conditioned}$-$\mathsf{NULL}_k$ and $(\kappa,\varepsilon)$-$\mathsf{conditioned}$-$\mathsf{FSLE}$, respectively, with $\kappa(n),\varepsilon(n)^{-1} = \poly(n)$.
Furthermore, we write $\NULL$ without an index if $k$ is given as part of the input.
We define $\kappa\cNULL$, and $\pcNULL$ in the natural way.

\begin{remark}
    The above problems $\NULL$ and $\FSLE$ are formulated for matrices with integer entries from $\mathbb Z$ in \cite{allender_complexity_1999}. They remark that their computational complexities remain unchanged when one allows rational entries from $\mathbb Q$ instead.
    We make the further observation that their complexity further remains unchanged if one allows Gaussian rationals $\mathbb Q(i)$, i.e. complex numbers of the form $(a+ib)/c$ for integers $a,b,c$.
    The reason is that we can use the following realification of a complex matrix $M = A+iB$ with $A,B\in\mathbb{Q}^{n\times n}$:
    \[
        M = A+iB
        \quad \mapsto \quad
        \mathcal{R}(M):=
        \begin{pmatrix}
            A & -B\\
            B & A
        \end{pmatrix}.
    \]
    Then
    \[
        \rank_{\mathbb{R}}(\mathcal{R}(M)) = 2\rank_{\mathbb{C}}(M),
    \]
    and, for any complex vector $b=s+it \in \mathbb Q(i)^n$,
    \[
        Mx=b \quad \text{if and only if} \quad
        \mathcal{R}(M)\begin{pmatrix}
            u\\v
        \end{pmatrix}
        = \begin{pmatrix}
            s\\
            t
        \end{pmatrix},
    \]
    where $u=\Re(x)$ and $v=\Im(x)$.

    Throughout, we always assume that the entries of the input matrices are complex rationals of polynomial bit complexity.
\end{remark}

\subsubsection{Classical upper bounds for nullity testing and feasibility}\label{sec:classical-upper-bounds}
We use ordinary subspace designs and kernel peeling to give an alternative reduction from $\mathsf{NULL}$ to $\mathsf{SINGULAR}$. This recovers the classical upper bounds of Allender, Beals, and Ogihara \cite{allender_complexity_1999} for these problems, whose original proof uses a sequence of
reductions including Berkowitz's algorithm \cite{BERKOWITZ1984147}.

Recall that a language $L$ belongs to $\mathsf{C_=L}$ if there exists a nondeterministic logspace machine~$M$ such that, for every input $x$, $x\in L$ if and only if $M$ has equally many accepting and rejecting computation paths on $x$.
We use the following standard facts.
\begin{fact}[\cite{Allender_Ogihara_94,Mahajan_Vinay_97,allender_complexity_1999}]
    $\SINGULAR$ is $\mathsf{C_=L}$-complete under deterministic logspace many-one reductions, and $\mathsf{C_=L}$ is closed under deterministic logspace conjunctive truth-table reductions.
\end{fact}
Here a conjunctive truth-table reduction from a language (or promise problem) $A$ to $B$ outputs a polynomial-length list of queries to $B$ so that an input is in $A$ if and only if every query is a yes-instance. We write $A\preccurlyeq_{\mathrm{ctt}}^{\mathsf L} B$.

\begin{lemma}\label{lem:null-reduces-to-singular}
    $\NULL \preccurlyeq^{\mathsf L}_{\mathrm{ctt}}\SINGULAR$.
\end{lemma}
\begin{proof}
    Let $(A,k)$ be an instance of $\NULL$, where $A\in\mathbb C^{n\times n}$ is the matrix which we wish to decide whether $\dim\ker(A)\geq k$.
    For each $1\leq k'<k$, take the Wronskian maps
    \[
        Q_1^{(k')},\ldots,Q_{m_{k'}}^{(k')}\in\mathbb C^{k'\times n}
    \]
    from \cref{thm:wronskian-subspace-design}, evaluated at $m_{k'}=k'(n-k')+1$ distinct points $T_{k'}=\{0,\ldots,m_{k'}-1\}$.
    For $k'=0$, take the zero map.
    Output singularity queries for each matrix
    \[
        A+\bigl(Q_i^{(k')}\bigr)^*Q_j^{(k')},
        \qquad 0\leq k'<k,\quad i,j\in[m_{k'}].
    \]
    There are polynomially many queries, and their entries are computable in deterministic logspace by the logspace uniformity of the Wronskian construction and since matrix multiplication is in deterministic logspace.

    If $\dim\ker(A)\geq k$, every perturbation has rank at most $k'<k$, and hence
    \[
        \rank\!\left(A+\bigl(Q_i^{(k')}\bigr)^*Q_j^{(k')}\right)
        \leq \rank(A)+k'
        \leq n-k+k'<n.
    \]
    Thus every query is singular.
    Conversely, if $1\leq d:=\dim\ker(A)<k$, the kernel peeling argument in \cref{lem:kernel-peeling}, applied with $k'=d$, yields indices $i,j$ for which
    \[
        \dim\ker\! \left( A+\bigl(Q_i^{(d)}\bigr)^*Q_j^{(d)} \right) =0.
    \]
    For $d=0$, the query $A$ itself is nonsingular.
    Thus at least one query is nonsingular, proving the reduction.
\end{proof}

Consequently, $\NULL\in\mathsf{C_=L}$, recovering \cite[Proposition~2.5]{allender_complexity_1999}.
This also gives an alternative route to $\FSLE\in\mathsf L^{\mathsf{C_=L}}$ through \cite[Corollary 2.15]{allender_complexity_1999}: for a linear system $Ax=b$,
\[
    Ax=b\text{ is feasible}
    \quad\Longleftrightarrow\quad
    \rank(A)=\rank([A\; b]).
\]
Both ranks can be determined in deterministic logspace using nullity queries.

\subsubsection{Quantum upper bounds for nullity testing and feasibility}
\label{sec:quantum-upper-bounds}

We first apply randomized robust subspace designs to conditioned nullity testing. The reduction has one-sided error: on a yes-instance all oracle queries are singular, while on a no-instance a successful design produces an invertible query.

\begin{theorem}\label{thm:pc-null-in-corl-to-qumal1}
    $\pcNULL\in\mathsf{coRL}^{\mathsf{Q_U MAL}_1}$.
\end{theorem}
\begin{proof}
    Similar to the proof in \cref{thm:fewqumal1-in-rl-to-uqumal1}, we consider an auxiliary problem $\pcSINGULAR'$. Its input consists of a positive semi-definite matrix $H\succeq0$ and a matrix $Q\in\C^{k'\times n}$, where $0\leq k'<n$, all entries of $Q$ are bounded by $|Q(i,j)| \leq 1$ and $\|H\|_{\op}\leq1$.
    Writing $M=H+Q^*Q$, the yes-instances have $\dim\ker M\geq1$, and the no-instances have $\lambda_{\min}(M)\geq1/\poly(n)$.
    This problem belongs to $\mathsf{Q_U MAL}_1$: given $(H,Q)$ on a read-only input tape, the entries of
    \[
        \widehat M:=\frac{H+Q^*Q}{1+nk'}
    \]
    are computable in deterministic logspace. The matrix $\widehat M$ is a contraction, has the same kernel as $M$, and is inverse-polynomially conditioned on no-instances.

    Let $(A,k)$ be a promised input to $\pcNULL$.
    Put $H=AA^*$.
    Then $H\succeq0$, $\|H\|_{\op}\leq1$, and $\dim\ker H=\dim\ker A$. If $k=1$, make the single oracle query $(H,0)$. Otherwise set $k'=k-1$ and use
    \cref{thm:finite-precision-design} to sample a robust subspace design
    \[
        \mathcal Q^{(k')}=
        \{Q_1^{(k')},\ldots,Q_m^{(k')}\}
        \subseteq\C^{k'\times n},
    \]
    of polynomial size and robustness $\eta=\Omega(1/\sqrt{nk'})$ with probability at least $2/3$. Query $\pcSINGULAR'$ on $(H,Q_i^{(k')})$ for each $i$, except replace entries of modulus larger than $1$ by $0$.
    Accept if all query answers are yes.
    Since entries of $H$ are computable in deterministic logspace from $A$, and the $Q_i^{(k')}$s are streamed to the oracle tape, the total space usage of the procedure is $O(\log n)$.

    If $\dim\ker A\geq k$, every query is singular, since
    \[
        \rank\bigl(H+(Q_i^{(k')})^*Q_i^{(k')}\bigr)
        \leq\rank H+k'\leq n-k+(k-1)<n.
    \]
    This also covers $k=1$ with $Q=0$. Hence the algorithm accepts yes-instances with certainty, regardless of the sampled subspace design.

    If $(A,k)$ is a no-instance, its promise gives $\lambda_k(H)=\sigma_k(A)^2\geq1/\poly(n)$.
    For $k=1$, $(H,0)$ is a promised no-query. For $k\geq2$, let $H_1$ contain the first $k'$ terms in a spectral decomposition of $H$, and put $H_2=H-H_1$. Then $H_1\succeq0$, $\dim\ker H_2=k'$, and $\lambda_{\min}^+(H_2)=\lambda_k(H)$. On the successful design event, the Hermitian case of \cref{lem:kernel-peeling,rem:psd-kernel-peeling} yields an index $i$ for which
    \[
        \lambda_{\min}
        \bigl(H_2+(Q_i^{(k')})^*Q_i^{(k')}\bigr)
        \geq\frac{\eta^2\lambda_k(H)}{3}
        \geq\frac1{\poly(n)}.
    \]
    Adding $H_1\succeq0$ preserves this lower bound. Thus one query is a promised no-instance, and the algorithm rejects.
    Answers to other queries outside the promise cannot change this outcome. The rejection probability is at least $2/3$.
\end{proof}

We next transfer this result to well-conditioned linear-system feasibility. The following rank-parity construction adapts the classical reduction of \cite[Corollary 2.15]{allender_complexity_1999}. We add the analysis of the spectral gap for our setting.

\begin{lemma}\label{lem:pcfsle-reduces-to-pcnull}
    Let $\kappa\geq1$ and $\varepsilon\in(0,1]$.
    Given an instance $(A,b)$ of $(\kappa,\varepsilon)$-$\mathsf{conditioned}$-$\mathsf{FSLE}$, one can construct in deterministic logspace a Hermitian contraction $H\in\C^{2n\times2n}$ such that
    \[
        Ax=b\text{ is feasible}
        \quad\Longleftrightarrow\quad
        \rank H\text{ is even},
        \qquad
        \lambda_{\min}^+(H)
        \geq\frac{\varepsilon^2}{18\kappa^2}.
    \]
\end{lemma}
\begin{proof}
    Put $T=[A\;b]$ and
    \[
        H:=\frac12
        \begin{pmatrix}
            AA^*&0\\
            0&TT^*
        \end{pmatrix}
        =\frac12
        \begin{pmatrix}
            AA^*&0\\
            0&AA^*+bb^*
        \end{pmatrix}.
    \]
    Its entries are computable exactly in deterministic logspace, and $\|H\|_{\op}\leq1$ because $\|A\|_{\op},\|b\|\leq1$.
    If $\rank A=r$, then $\rank T=r$ when $Ax=b$ is feasible and $\rank T=r+1$ otherwise.
    Hence $\rank H=r+\rank T$ has the claimed parity.

    It remains to bound the smallest non-zero eigenvalue. Since $\sigma_{\min}^+(A)\geq1/\kappa$, it suffices to prove
    \[
        \sigma_{\min}^+(T)\geq\frac{\varepsilon}{3\kappa},
    \]
    since $\lambda_{\min}^+(H) = \frac{1}{2} \min\{\left(\splus(A)\right)^2, \left(\splus(T)\right)^2\}$.

    If $Ax=b$ is feasible, then $b\in\im A$. The matrices $TT^*=AA^*+bb^*$ and $AA^*$ have the same kernel, and $TT^*\succeq AA^*$. Therefore $\sigma_{\min}^+(T)\geq\sigma_{\min}^+(A) \geq1/\kappa$.

    Otherwise, if $Ax=b$ is infeasible, decompose $b=b_{\parallel}+b_{\perp}$ with $b_{\parallel}\in\im A$ and $b_{\perp}\perp\im A$. The input promise gives $\|b_{\perp}\|\geq\varepsilon$.
    For $y\in\im T=\im A\oplus\operatorname{span}\{b_{\perp}\}$, set
    \[
        t:=\frac{\langle b_{\perp},y\rangle}{\|b_{\perp}\|^2},
        \qquad
        x:=A^+(y-tb),
    \]
    where $A^+$ is the Moore-Penrose pseudoinverse, which maps every vector $z\in \im(A)$ to its minimum norm preimage $x$ such that $Ax=z$.
    Then $T(x,t)=Ax+tb=y$. Moreover,
    \[
        |t|\leq\frac{\|y\|}{\varepsilon},
        \qquad
        \|x\|\leq\kappa\|y-tb\|
        \leq\kappa\left(1+\frac1\varepsilon\right)\|y\|,
    \]
    where we used that $\|A^+\|_{\op} \leq \kappa$.
    Thus $\|(x,t)\|\leq3\kappa\|y\|/\varepsilon$.
    The minimum-norm preimage of $y$ has no larger norm, so $\|T^+\|_{\op}\leq3\kappa/\varepsilon$, proving the claimed bound on $\sigma_{\min}^+(T)$ and hence on $H$.
\end{proof}

For every threshold $k\in\{1,\ldots,2n\}$, the query ``$\dim\ker H\geq k$'' is a promised instance of $\widetilde\kappa$-$\mathsf{conditioned}$-$\mathsf{NULL}_{k}$ with $\widetilde\kappa=18\kappa^2/\varepsilon^2$: if the answer is no, its $k$th smallest singular value is at least $\lambda_{\min}^+(H)$.
These queries determine $\rank H$ and therefore the feasibility of $Ax=b$.

\begin{theorem}\label{thm:pcfsle-upper-bounds}
    \label{thm:pcfsle-in-bpl-to-qumal1}
    $\pcFSLE\in\mathsf{BPL}^{\mathsf{Q_U MAL}_1}$.
\end{theorem}
\begin{proof}
    Amplify the algorithm of \cref{thm:pc-null-in-corl-to-qumal1} to error smaller than $1/(10n)$ on each threshold query. A union bound then lets us recover $\rank H$ with bounded error.
\end{proof}

\section{Acknowledgements}

The authors gratefully acknowledge stimulating discussions with Quinten Tupker on the topic of logspace complexity, as well as with Isabella Ziccardi about Gaussian sampling.
This work has received support under the program ``Investissement d'Avenir" launched by the French Government and implemented by ANR, with the reference “ANR-22-CMAS-0001, QuanTEdu-France”.
SA was supported in part by the European QuantERA project QOPT (ERA-NET Cofund 2022-25), the French PEPR integrated projects EPiQ (ANR-22-PETQ0007) and HQI (ANR-22-PNCQ-0002), and the French ANR project QUOPS (ANR-22-CE47-0003-01).

\paragraph{AI disclosure:}

The main direction and conceptual contributions (notion of robust subspace designs, kernel peeling and the complexity-theoretic applications) were conceived by the authors.
GPT 5.6 Sol came up with the Grassmannian net argument for the randomized construction and with an improved version of our explicit robust construction (improving the robustness parameter from inverse exponential to inverse polynomial in \cref{thm:coordinate}).
GPT 6.0 Astra came up with the idea for improving the size from $\binom{n}{k}$ to $\poly(n) k^{O(k)}$ through sketching in \cref{thm:composed-subspace-design}.
GPT 5.6 Sol, 6.0 Astra and 6.1 Sol were also used to improve the presentation of parts of the manuscript and to help identify errors.
The authors take full responsibility for the content and correctness of the manuscript.

\bibliographystyle{alpha}
\bibliography{biblio}

@article{ValiantV86,
  author        = {Leslie G. Valiant and Vijay V. Vazirani},
  title         = {{NP} is as Easy as Detecting Unique Solutions},
  journal       = {Theor. Comput. Sci.},
  volume        = {47},
  number        = {3},
  pages         = {85--93},
  year          = {1986},
  doi           = {10.1016/0304-3975(86)90135-0},
  url           = {https://doi.org/10.1016/0304-3975(86)90135-0},
}

@inproceedings{DPT24,
  author        = {Dean Doron and Edward Pyne and Roei Tell},
  title         = {Opening Up the Distinguisher: A Hardness to Randomness Approach for {BPL} = {L} That Uses Properties of {BPL}},
  booktitle     = {Proceedings of the 56th Annual ACM Symposium on Theory of Computing (STOC)},
  pages         = {2039--2049},
  publisher     = {ACM},
  year          = {2024},
  doi           = {10.1145/3618260.3649772},
  note          = {See Claim 3.17 in the ECCC version, Report TR23-208},
}

@article{ABI86,
  author        = {Noga Alon and L{\'a}szl{\'o} Babai and Alon Itai},
  title         = {A Fast and Simple Randomized Parallel Algorithm for the Maximal Independent Set Problem},
  journal       = {Journal of Algorithms},
  volume        = {7},
  number        = {4},
  pages         = {567--583},
  year          = {1986},
  doi           = {10.1016/0196-6774(86)90019-2},
}

@inproceedings{CW09,
  author        = {Kenneth L. Clarkson and David P. Woodruff},
  title         = {Numerical Linear Algebra in the Streaming Model},
  booktitle     = {Proceedings of the 41st Annual ACM Symposium on Theory of Computing (STOC)},
  pages         = {205--214},
  publisher     = {ACM},
  year          = {2009},
  doi           = {10.1145/1536414.1536445},
}

@inproceedings{Sarlos06,
  author        = {Tam{\'a}s Sarl{\'o}s},
  title         = {Improved Approximation Algorithms for Large Matrices via Random Projections},
  booktitle     = {Proceedings of the 47th Annual IEEE Symposium on Foundations of Computer Science (FOCS)},
  pages         = {143--152},
  publisher     = {IEEE Computer Society},
  year          = {2006},
  doi           = {10.1109/FOCS.2006.37},
}

@article{DesrosiersForrester2006,
  author        = {Patrick Desrosiers and Peter J. Forrester},
  title         = {{Asymptotic correlations for Gaussian and Wishart matrices with external source}},
  journal       = {International Mathematics Research Notices},
  volume        = {2006},
  pages         = {27395},
  year          = {2006},
  doi           = {10.1155/IMRN/2006/27395},
  eprint        = {math-ph/0604012},
  archivePrefix = {arXiv},
}

@inproceedings{BrakensiekChenDharZhang2026,
  author        = {Joshua Brakensiek and Yeyuan Chen and Manik Dhar and Zihan Zhang},
  title         = {From Random to Explicit via Subspace Designs with Applications to Local Properties and Matroids},
  booktitle     = {Proceedings of the 58th Annual ACM Symposium on Theory of Computing},
  series        = {STOC 2026},
  publisher     = {Association for Computing Machinery},
  year          = {2026},
  doi           = {10.1145/3798129.3800778},
}

@techreport{Applebaum2026,
  author        = {Benny Applebaum},
  title         = {From Nondeterministic to Modular Branching Programs, Uniformly and in Logspace},
  number        = {TR26--131},
  institution   = {Electronic Colloquium on Computational Complexity},
  year          = {2026},
}

@inproceedings{FeffermanRemscrim2021,
  author        = {Bill Fefferman and Zachary Remscrim},
  title         = {Eliminating Intermediate Measurements in Space-Bounded Quantum Computation},
  booktitle     = {Proceedings of the 53rd Annual ACM SIGACT Symposium on Theory of Computing},
  series        = {STOC 2021},
  pages         = {1343--1356},
  publisher     = {Association for Computing Machinery},
  year          = {2021},
  doi           = {10.1145/3406325.3451051},
  eprint        = {2006.03530},
  archivePrefix = {arXiv},
  primaryClass  = {quant-ph},
}

@inproceedings{tashma,
  author        = {Ta-Shma, Amnon},
  title         = {Inverting well conditioned matrices in quantum logspace},
  booktitle     = {Proceedings of the Annual ACM Symposium on Theory of Computing},
  pages         = {881--890},
  year          = {2013},
  month         = {06},
  doi           = {10.1145/2488608.2488720},
}

@inproceedings{fefferman_lin,
  author        = {Fefferman, Bill and Lin, Cedric Yen-Yu},
  title         = {A Complete Characterization of Unitary Quantum Space},
  editor        = {Karlin, Anna R.},
  booktitle     = {9th Innovations in Theoretical Computer Science Conference (ITCS 2018)},
  series        = {LIPIcs},
  volume        = {94},
  pages         = {4:1--4:21},
  publisher     = {Schloss Dagstuhl--Leibniz-Zentrum fuer Informatik},
  address       = {Dagstuhl, Germany},
  year          = {2018},
  doi           = {10.4230/LIPIcs.ITCS.2018.4},
  url           = {https://doi.org/10.4230/LIPIcs.ITCS.2018.4},
}

@article{allender_complexity_1999,
  author        = {Eric Allender and Robert Beals and Mitsunori Ogihara},
  title         = {The Complexity of Matrix Rank and Feasible Systems of Linear Equations},
  journal       = {Computational Complexity},
  volume        = {8},
  number        = {2},
  pages         = {99--126},
  year          = {1999},
  doi           = {10.1007/s000370050023},
}

@article{JKKSSZ09,
  author        = {Jain, Rahul and Kerenidis, Iordanis and Kuperberg, Greg and Santha, Miklos and Sattath, Or and Zhang, Shengyu},
  title         = {On the Power of a Unique Quantum Witness},
  journal       = {Theory of Computing},
  volume        = {8},
  number        = {17},
  pages         = {375--400},
  publisher     = {Theory of Computing},
  year          = {2012},
  doi           = {10.4086/toc.2012.v008a017},
  url           = {https://theoryofcomputing.org/articles/v008a017/},
}

@article{cook_taxonomy_1985,
  author        = {Cook, Stephen A.},
  title         = {A taxonomy of problems with fast parallel algorithms},
  journal       = {Information and Control},
  volume        = {64},
  number        = {1--3},
  pages         = {2--22},
  year          = {1985},
  month         = jan,
  doi           = {10.1016/S0019-9958(85)80041-3},
  url           = {https://linkinghub.elsevier.com/retrieve/pii/S0019995885800413},
  issn          = {00199958},
}

@inproceedings{Allender_Ogihara_94,
  author        = {Allender, Eric and Ogihara, Mitsunori},
  title         = {{Relationships among PL, {\#}L, and the determinant}},
  booktitle     = {Proceedings of IEEE 9th Annual Conference on Structure in Complexity Theory},
  pages         = {267--278},
  year          = {1994},
  doi           = {10.1109/SCT.1994.315797},
}

@article{Mahajan_Vinay_97,
  author        = {Mahajan, Meena and Vinay, V.},
  title         = {Determinant: Combinatorics, algorithms, and complexity},
  journal       = {Chicago Journal of Theoretical Computer Science},
  volume        = {1997},
  number        = {5},
  publisher     = {MIT Press},
  year          = {1997},
}

@article{hough2006determinantal,
  author        = {Hough, J. Ben and Krishnapur, Manjunath and Peres, Yuval and Vir{\'a}g, B{\'a}lint},
  title         = {Determinantal processes and independence},
  journal       = {Probability Surveys},
  volume        = {3},
  pages         = {206--229},
  publisher     = {The Institute of Mathematical Statistics and the Bernoulli Society},
  year          = {2006},
  doi           = {10.1214/154957806000000078},
}

@book{Forrester2010,
  author        = {Forrester, Peter J.},
  title         = {Log-Gases and Random Matrices},
  series        = {London Mathematical Society Monographs},
  volume        = {34},
  publisher     = {Princeton University Press},
  address       = {Princeton, NJ},
  year          = {2010},
  isbn          = {9780691128290},
}

@article{chatterjee2026bipartite,
  author        = {Chatterjee, Abhranil and Ghosh, Sumanta and Gurjar, Rohit and Raj, Roshan and Thierauf, Thomas},
  title         = {Bipartite Matching is in {NC}},
  journal       = {Electronic Colloquium on Computational Complexity (ECCC)},
  volume        = {TR26-100},
  year          = {2026},
  url           = {https://eccc.weizmann.ac.il/report/2026/100/},
}

@article{GuruswamiKopparty2016ExplicitSubspaceDesigns,
  author        = {Venkatesan Guruswami and Swastik Kopparty},
  title         = {Explicit Subspace Designs},
  journal       = {Combinatorica},
  volume        = {36},
  number        = {2},
  pages         = {161--185},
  year          = {2016},
  doi           = {10.1007/s00493-014-3179-7},
  eprint        = {1304.4205},
  archivePrefix = {arXiv},
  primaryClass  = {cs.CC},
}

@inproceedings{GuruswamiXing2013SubcodesSingleton,
  author        = {Venkatesan Guruswami and Chaoping Xing},
  title         = {{List decoding Reed--Solomon, algebraic-geometric, and Gabidulin subcodes up to the singleton bound}},
  booktitle     = {Proceedings of the 45th Annual ACM Symposium on Theory of Computing (STOC)},
  pages         = {843--852},
  publisher     = {ACM},
  year          = {2013},
  doi           = {10.1145/2488608.2488715},
  eprint        = {1211.0865},
  archivePrefix = {arXiv},
  primaryClass  = {cs.IT},
}

@article{Szarek1998MetricEntropyHomogeneousSpaces,
  author        = {Stanis{\l}aw J. Szarek},
  title         = {Metric Entropy of Homogeneous Spaces},
  journal       = {Banach Center Publications},
  volume        = {43},
  pages         = {395--410},
  year          = {1998},
  doi           = {10.4064/-43-1-395-410},
  eprint        = {math/9701213},
  archivePrefix = {arXiv},
}

@book{Vershynin2018HighDimensionalProbability,
  author        = {Roman Vershynin},
  title         = {High-Dimensional Probability: An Introduction with Applications in Data Science},
  series        = {Cambridge Series in Statistical and Probabilistic Mathematics},
  volume        = {47},
  publisher     = {{Cambridge University Press}},
  year          = {2018},
  doi           = {10.1017/9781108231596},
}

@book{Bhatia1997MatrixAnalysis,
  author        = {Rajendra Bhatia},
  title         = {Matrix Analysis},
  series        = {Graduate Texts in Mathematics},
  volume        = {169},
  publisher     = {Springer},
  year          = {1997},
  doi           = {10.1007/978-1-4612-0653-8},
}

@misc{jain2011poweruniquequantumwitness,
  author        = {Rahul Jain and Iordanis Kerenidis and Greg Kuperberg and Miklos Santha and Or Sattath and Shengyu Zhang},
  title         = {On the power of a unique quantum witness},
  year          = {2011},
  url           = {https://arxiv.org/abs/0906.4425},
  eprint        = {0906.4425},
  archivePrefix = {arXiv},
  primaryClass  = {quant-ph},
}

@incollection{goreinov2010find,
  author        = {Goreinov, Sergei A and Oseledets, Ivan V and Savostyanov, Dimitry V and Tyrtyshnikov, Eugene E and Zamarashkin, Nikolay L},
  title         = {How to find a good submatrix},
  booktitle     = {Matrix Methods: Theory, Algorithms And Applications: Dedicated to the Memory of Gene Golub},
  pages         = {247--256},
  publisher     = {World Scientific},
  year          = {2010},
}

@inproceedings{mangoubi2022private,
  author        = {Mangoubi, Oren and Wu, Yikai and Kale, Satyen and Thakurta, Abhradeep and Vishnoi, Nisheeth K},
  title         = {Private matrix approximation and geometry of unitary orbits},
  booktitle     = {Conference on Learning Theory},
  pages         = {3547--3588},
  organization  = {PMLR},
  year          = {2022},
}

@misc{DLMF,
  author        = {{NIST Digital Library of Mathematical Functions}},
  title         = {Chapter 18: Orthogonal Polynomials, Section 18.14: Inequalities},
  publisher     = {National Institute of Standards and Technology},
  year          = {2025},
  note          = {Eq.~(18.14.8), Release 1.2.4 of 2025-03-15},
  howpublished  = {\url{https://dlmf.nist.gov/18.14}},
}

@book{anderson2010introduction,
  author        = {Anderson, Greg W and Guionnet, Alice and Zeitouni, Ofer},
  title         = {An introduction to random matrices},
  series        = {Cambridge Studies in Advanced Mathematics},
  volume        = {118},
  publisher     = {Cambridge University Press},
  year          = {2010},
}

@book{aubrun2017alice,
  author        = {Aubrun, Guillaume and Szarek, Stanis{\l}aw J},
  title         = {{Alice} and {Bob} meet {Banach}},
  volume        = {223},
  publisher     = {American Mathematical Society},
  year          = {2017},
}

@inproceedings{ben2008quantum,
  author        = {Ben-Aroya, Avraham and Schwartz, Oded and Ta-Shma, Amnon},
  title         = {Quantum expanders: Motivation and constructions},
  booktitle     = {2008 23rd Annual IEEE Conference on Computational Complexity},
  pages         = {292--303},
  organization  = {IEEE},
  year          = {2008},
}

@article{li2025linear,
  author        = {Li, Yinan and Qiao, Youming and Wigderson, Avi and Wigderson, Yuval and Zhang, Chuanqi},
  title         = {On Linear-Algebraic Notions of Expansion},
  journal       = {Theory of Computing},
  volume        = {21},
  number        = {1},
  pages         = {1--28},
  year          = {2025},
  doi           = {10.4086/toc.2025.v021a001},
}

@article{lubotzky2008dimension,
  author        = {Lubotzky, Alexander and Zelmanov, Efim},
  title         = {Dimension expanders},
  journal       = {Journal of Algebra},
  volume        = {319},
  number        = {2},
  pages         = {730--738},
  publisher     = {Elsevier},
  year          = {2008},
}

@inproceedings{forbes2015dimension,
  author        = {Forbes, Michael A. and Guruswami, Venkatesan},
  title         = {Dimension Expanders via Rank Condensers},
  booktitle     = {Approximation, Randomization, and Combinatorial Optimization. Algorithms and Techniques (APPROX/RANDOM 2015)},
  pages         = {800--814},
  publisher     = {Schloss Dagstuhl--Leibniz-Zentrum fuer Informatik},
  year          = {2015},
  doi           = {10.4230/LIPIcs.APPROX-RANDOM.2015.800},
}

@article{CDL01,
  author        = {Chiu, Andrew and Davida, George and Litow, Bruce},
  title         = {Division in logspace-uniform $\mbox{NC}^1$},
  journal       = {RAIRO. Theoretical Informatics and Applications},
  volume        = {35},
  number        = {3},
  pages         = {259--275},
  publisher     = {EDP-Sciences},
  year          = {2001},
  url           = {https://www.numdam.org/item/ITA_2001__35_3_259_0/},
  mrnumber      = {1869217},
  zbl           = {1014.68062},
}

@book{KitaevShenVyalyi2002,
  author        = {Alexei Yu. Kitaev and Alexander Shen and Mikhail N. Vyalyi},
  title         = {Classical and Quantum Computation},
  series        = {Graduate Studies in Mathematics},
  volume        = {47},
  publisher     = {American Mathematical Society},
  year          = {2002},
  doi           = {10.1090/gsm/047},
}

@article{MZ12,
  author        = {David M. Mason and Harrison H. Zhou},
  title         = {{Quantile coupling inequalities and their applications}},
  journal       = {Probability Surveys},
  volume        = {9},
  pages         = {439--479},
  publisher     = {Institute of Mathematical Statistics and Bernoulli Society},
  year          = {2012},
  doi           = {10.1214/12-PS198},
  url           = {https://doi.org/10.1214/12-PS198},
}

@inproceedings{fefferman_et_al:LIPIcs.ICALP.2016.14,
  author        = {Fefferman, Bill and Kobayashi, Hirotada and Lin, Cedric Yen-Yu and Morimae, Tomoyuki and Nishimura, Harumichi},
  title         = {{Space-Efficient Error Reduction for Unitary Quantum Computations}},
  editor        = {Chatzigiannakis, Ioannis and Mitzenmacher, Michael and Rabani, Yuval and Sangiorgi, Davide},
  booktitle     = {43rd International Colloquium on Automata, Languages, and Programming (ICALP 2016)},
  series        = {Leibniz International Proceedings in Informatics (LIPIcs)},
  volume        = {55},
  pages         = {14:1--14:14},
  publisher     = {Schloss Dagstuhl -- Leibniz-Zentrum f{\"u}r Informatik},
  address       = {Dagstuhl, Germany},
  year          = {2016},
  doi           = {10.4230/LIPIcs.ICALP.2016.14},
  url           = {https://drops.dagstuhl.de/entities/document/10.4230/LIPIcs.ICALP.2016.14},
  isbn          = {978-3-95977-013-2},
  issn          = {1868-8969},
  urn           = {urn:nbn:de:0030-drops-62975},
}

@inproceedings{All86,
  author        = {Allender, Eric W.},
  title         = {The complexity of sparse sets in {P}},
  editor        = {Selman, Alan L.},
  booktitle     = {Structure in Complexity Theory},
  pages         = {1--11},
  publisher     = {Springer Berlin Heidelberg},
  address       = {Berlin, Heidelberg},
  year          = {1986},
  isbn          = {978-3-540-39825-7},
}

@inproceedings{BJLR91,
  author        = {Buntrock, Gerhard and Jenner, Birgit and Lange, Klaus-J{\"o}rn and Rossmanith, Peter},
  title         = {Unambiguity and fewness for logarithmic space},
  editor        = {Budach, L.},
  booktitle     = {Fundamentals of Computation Theory},
  pages         = {168--179},
  publisher     = {Springer Berlin Heidelberg},
  address       = {Berlin, Heidelberg},
  year          = {1991},
  isbn          = {978-3-540-38391-8},
}

@article{VALIANT197620,
  author        = {Leslie G. Valiant},
  title         = {Relative complexity of checking and evaluating},
  journal       = {Information Processing Letters},
  volume        = {5},
  number        = {1},
  pages         = {20--23},
  year          = {1976},
  doi           = {10.1016/0020-0190(76)90097-1},
  url           = {https://www.sciencedirect.com/science/article/pii/0020019076900971},
  issn          = {0020-0190},
}

@article{BERKOWITZ1984147,
  author        = {Stuart J. Berkowitz},
  title         = {On computing the determinant in small parallel time using a small number of processors},
  journal       = {Information Processing Letters},
  volume        = {18},
  number        = {3},
  pages         = {147--150},
  year          = {1984},
  doi           = {10.1016/0020-0190(84)90018-8},
  url           = {https://www.sciencedirect.com/science/article/pii/0020019084900188},
  issn          = {0020-0190},
}

@article{Aharonov_2022,
  author        = {Aharonov, Dorit and Ben-Or, Michael and Brandão, Fernando G.S.L. and Sattath, Or},
  title         = {The Pursuit of Uniqueness: Extending {Valiant}-{Vazirani} Theorem to the Probabilistic and Quantum Settings},
  journal       = {Quantum},
  volume        = {6},
  pages         = {668},
  publisher     = {Verein zur Forderung des Open Access Publizierens in den Quantenwissenschaften},
  year          = {2022},
  month         = mar,
  doi           = {10.22331/q-2022-03-17-668},
  url           = {http://dx.doi.org/10.22331/q-2022-03-17-668},
  issn          = {2521-327X},
}

@inproceedings{anshu_et_al:LIPIcs.ITCS.2026.10,
  author        = {Anshu, Anurag and Haferkamp, Jonas and Hwang, Yeongwoo and Nguyen, Quynh T.},
  title         = {{On the Complexity of Unique Quantum Witnesses and Quantum Approximate Counting}},
  editor        = {Saraf, Shubhangi},
  booktitle     = {17th Innovations in Theoretical Computer Science Conference (ITCS 2026)},
  series        = {Leibniz International Proceedings in Informatics (LIPIcs)},
  volume        = {362},
  pages         = {10:1--10:20},
  publisher     = {Schloss Dagstuhl -- Leibniz-Zentrum f{\"u}r Informatik},
  address       = {Dagstuhl, Germany},
  year          = {2026},
  doi           = {10.4230/LIPIcs.ITCS.2026.10},
  url           = {https://drops.dagstuhl.de/entities/document/10.4230/LIPIcs.ITCS.2026.10},
  isbn          = {978-3-95977-410-9},
  issn          = {1868-8969},
  urn           = {urn:nbn:de:0030-drops-252978},
}

\appendix
\crefalias{section}{appendix}

\section{Explicit construction based on Wronskian evaluation maps}\label{sec:wronskian-subspace}

\begin{lemma}[Wronskian degree bound]\label{lem:wronskian}
    Let $N\in\C^{n\times k}$ have rank $k$. Then
    \[
        W_N(z):=\det(Q_{k,z}N)
    \]
    is a nonzero polynomial of degree at most $k(n-k)$.
\end{lemma}

\begin{proof}
    Let $f_1,\ldots,f_k\in\C[z]_{<n}$ be the polynomials corresponding to the columns of $N$. Since $N$ has rank $k$, the polynomials $f_1,\ldots,f_k$ are linearly independent.
    By an invertible change of basis, we may assume that
    \[
        \deg(f_1)>\deg(f_2)>\cdots>\deg(f_k).
    \]
    Indeed, replacing $N$ by $NR$ for some invertible matrix $R\in\mathbb C^{k\times k}$ changes $W_N$ only by the nonzero scalar $\det(R)$.
    Writing $d_j:=\deg(f_j)$, we thus have
    \[
        d_j \leq n-j.
    \]
    Now
    \[
        Q_{k,z}N
        =
        \begin{pmatrix}
            f_1(z) & \cdots & f_k(z)\\
            f_1'(z) & \cdots & f_k'(z)\\
            \vdots & & \vdots\\
            f_1^{(k-1)}(z) & \cdots & f_k^{(k-1)}(z)
        \end{pmatrix},
    \]
    so $W_N$ is the Wronskian of $f_1,\ldots,f_k$.
    Over a field of characteristic zero, the Wronskian of linearly independent polynomials is nonzero. Finally, as $\deg(f_j^{(\ell)}) \leq d_j - \ell$, we find that
    \[
        \deg W_N
        \leq \sum_{j=1}^k d_j-\frac{k(k-1)}2
        \leq \sum_{j=1}^k (n-j)-\frac{k(k-1)}2
        =k(n-k).\qedhere
    \]
\end{proof}

With the above lemma we can now prove \cref{thm:wronskian-subspace-design}.
\wronskiandesign*
\begin{proof}
    Fix $1\leq k<n$, a set $T_k\subseteq \mathbb R$ of $m$ distinct points and the family $\mathcal{H}_{n,k}$ as in the statement.
    Let $W\subseteq \mathbb C^n$ be a $k$-dimensional subspace, and let $N\in\C^{n\times k}$ have columns forming a basis of $W$.
    By \cref{lem:wronskian}, $W_N(z) = \det(Q_{k,z}N)$ is a nonzero polynomial of degree
    \[
        \deg W_N \leq k(n-k).
    \]
    For $a\in T_k$, let
    \[
        k_a := \dim(H_a\cap W) = \dim\ker(Q_{k,a}N).
    \]
    Since we have $\rank(Q_{k,a}N)=k-k_a$, there exist invertible matrices $P,Q\in\mathbb C^{k\times k}$ such that
    \[
        P (Q_{k,a}N)Q
        =
        \begin{pmatrix}
            I_{k-k_a} & 0\\
            0 & 0
        \end{pmatrix},
    \]
    i.e. every entry in the last $k_a$ rows of $P (Q_{k,z}N)Q$ vanishes at $z=a$, and is therefore divisible by $(z-a)$. Factoring one copy of $(z-a)$ from each of these rows yields
    \[
        (z-a)^{k_a}\mid\det P (Q_{k,z}N)Q.
    \]
    On the other hand,
    \[
        \det P (Q_{k,z}N)Q = \det(P)W_N(z)\det(Q),
    \]
    where $\det(P),\;\det(Q)\neq0$. Consequently,
    \[
        (z-a)^{k_a}\mid W_N(z).
    \]
    Putting all together, we find
    \begin{align*}
        \sum_{a\in T_k}\dim(H_a\cap W)
        = \sum_{a\in T_k} k_a
        \leq \sum_{a\in T_k}\ord_a \left( W_N(z) \right)
        \leq \deg W_N(z)
        \leq k(n-k),
    \end{align*}
    where $\ord_a( W_N(z) )$ denotes the multiplicity of the root $a$ in $W_N(z)$, i.e. the maximum number of times $k\geq 0$ such that $(z-a)^k \mid W_N(z)$.
    Thus, $\mathcal H_{n,k}$ is a $(k,\;k(n-k))$-strong subspace design.

    The computability in space $O(\log nm)$ of the entries of the $Q_{k,a}$ given $T_k = \{0,\dots,m-1\}$ follows directly since integer arithmetic and matrix multiplication over polynomial-bit numbers can be performed in space logarithmic in the input size (see e.g. \cite{CDL01}).
\end{proof}

\section{Proofs of Gaussian matrix concentration} \label{sec:determinantal-process}

Let $G\in \C^{k\times k}$ be a square, standard complex Gaussian matrix; that is, its entries are independent complex Gaussian random variables with density $\frac{1}{\pi}e^{-|z|^2}$.
We devote this appendix to characterizing the smaller singular values of $G$.
Our discussion will largely follow \cite{Forrester2010,hough2006determinantal}.

\subsection{Background: Laguerre unitary ensemble and determinantal point processes}

If $G$ is a square, standard complex Gaussian matrix, then the random matrix $G G^*$ is drawn from the \emph{complex Wishart ensemble}.
Let $x_1,\ldots,x_k$ be the (unordered) eigenvalues of the positive semidefinite matrix $GG^*$.
Equivalently, $x_1,\ldots,x_k$ are the (unordered) squared singular values of $G$.
By \cite[Proposition 3.2.2]{Forrester2010}, the random variable ${\cal X} = (x_1,\dots,x_k)$ follows the square Laguerre unitary ensemble (LUE) with joint density $p(x_1,\dots,x_k)$ proportional to
\[
    p(x_1,\dots,x_k)
    \propto \prod_{1\leq i<j\leq k}(x_i-x_j)^2
    \prod_{i=1}^k e^{-x_i},
    \qquad x_1,\ldots,x_k\geq 0.
\]

The LUE ${\cal X} = (x_1,\dots,x_k)$ is a point process on the real half line $[0,\infty)$.
For any subset $D \subset [0,\infty)$, we denote by ${\cal X}(D)$ the number of points in $\cal X$ that fall in $D$, i.e.,
\[
    {\cal X}(D)
    = \left|\{j \in [k]: x_j \in D\}\right|.
\]
Following \cite{hough2006determinantal}, for any $1 \leq \ell \leq k$ we define the \emph{joint intensities} (or correlation functions) of $\cal X$ as functions $\rho_\ell:[0,\infty)^\ell \to [0,\infty)$ such that for any family of mutually disjoint subsets $D_1,\dots,D_\ell \subset [0,\infty)$ we have
\[
    \mathbb{E}\left[ \Pi_{j=1}^\ell {\cal X}(D_j) \right]
    = \int_{\Pi_{j=1}^\ell D_j} \rho_\ell(y_1,\dots,y_\ell) dy_1 \dots dy_\ell.
\]
In particular, when $\ell = k$ then $\rho_\ell = k!p$ while for $\ell = 1$ the intensity $\rho_1(x)$ simply corresponds to the eigenvalue density.

The key fact is that the LUE is a \emph{determinantal point process} (DPP). For general background, see \cite{Forrester2010}.
This means that there exists a kernel $K:[0,\infty)^2 \to \mathbb R$ such that
\[
    \rho_\ell(x_1,\dots,x_\ell)
    = \det(K(x_i,x_j)_{1 \leq i,j \leq \ell}).
\]
For the LUE, the kernel is
\[
    K(x,y)
    = e^{-(x+y)/2} \sum_{j=0}^{k-1} L_j(x) L_j(y),
\]
with $L_j(x)$ the $j$-th Laguerre polynomial, normalized so that
\[
    \int_0^\infty e^{-x}L_j(x)L_\ell(x)\,dx
    = \delta_{j\ell}.
\]
To explain the above facts more specifically, we refer to Section 2 of \cite{DesrosiersForrester2006}. The equation 2.3 is the key for us which states
\[
    K(x,y)=\sqrt{w(x)w(y)}\sum_{j=0}^{k-1}\frac{p_j(x) p_j(y)}{\norm{p_j}^2}.
\]
For us the weight functions are $w(x)=e^{-x}$ and $w(y)=e^{-y}$. Furthermore, as explained in Equation 2.7 of \cite{DesrosiersForrester2006},
\[
    p_j(x)=(-1)^j j! L_j(x)
    \qquad
    p_j(y)=(-1)^j j! L_j(y).
\]
Putting these values, we derive the required expression for the kernel entry $K(x,y)$ \footnote{For the square matrix case the term $\alpha$ used in \cite{DesrosiersForrester2006} is $0$.}.

By the standard inequality $e^{-x/2}\abs{L_k(x)} \le L_k(0) \le 1$ for $x \geq 0$ \cite{DLMF}, we get that
\[
    \rho_1(x)
    = K(x,x)
    \le k.
\]
It follows that for every $0 < u \leq 1$ we have
\begin{equation} \label{eq:mean-count}
    \mathbb{E}[{\cal X}([0,u^2/k])]
    = \int_0^{u^2/k} \rho_1(x) dx
    \leq \int_0^{u^2/k} k dx
    = u^2.
\end{equation}

\subsection{Singular value concentration}

To prove a concentration bound over families of Gaussian matrices, we will use the following theorem.
It is essentially \cite[Corollary 4.2.24]{anderson2010introduction}, or alternatively follows by applying \cite[Theorem 7]{hough2006determinantal} to the restricted DPP~${\cal X}|_D$ with kernel $K|_D(x,y) = K(x,y)$ if $(x,y) \in D$ and $K|_D(x,y) = 0$ otherwise.
\begin{theorem}
    Suppose $\cal X$ is a DPP on $[0,\infty)$ with a trace-class kernel $K$.
    Fix a measurable subset $D \subset [0,\infty)$ and consider the restricted kernel $K|_D$ with eigenvalues $\{\lambda_j(K|_D) \in [0,1]\}$.
    We have that
    \[
        {\cal X}(D)
        \overset{d}{=} \sum_j \mathrm{Bernoulli}(\lambda_j).
    \]
    I.e., ${\cal X}(D)$ is distributed as a sum of independent $\mathrm{Bernoulli}(\lambda_j)$ random variables.
\end{theorem}

We will apply this theorem to the LUE kernel.
To get a bound on the eigenvalues $\lambda_j(K|_D)$, we note that $K|_D$ is positive semi-definite and that
\[
    \sum_{j=1}^k \lambda_j(K|_D)
    = \mathrm{Tr}(K|_D)
    = \int_D K(x,x) dx
    \leq k |D|,
\]
with $|D|$ the Lebesgue measure of $D$.
This implies that ${\cal X}(D)$ will be distributed as a sum of $k$ Bernoulli with sum of parameters bounded by $k|D|$.
We use this to prove the following claim, which implies \cref{lem:small-count}.

\begin{claim} \label{lem:sing-value-count}
    Fix $0 < u \leq 1$.
    If ${\cal X}_1,\dots,{\cal X}_m$ are independent copies of the LUE ${\cal X}$, then
    \[
        \Prob\left[\sum_{i=1}^m {\cal X}_i([0,u^2/k]) \geq 2 u^2 m \right]
        \le e^{-u^2 m/3}.
    \]
\end{claim}
\begin{proof}
    Let $S = \sum_{i=1}^m {\cal X}_i([0,u^2/k])$.
    By the preceding argument, we get that $S$ behaves as a sum of $m k$ independent Bernoulli random variables with sum of parameter bounded by $u^2 m$, and so mean $\mu = \mathbb{E} S \leq u^2 m$.

    Let $1+\delta=\frac{2u^2m}{\mu}$.
    By the multiplicative Chernoff bound for independent Bernoulli variables, we know that for all $\delta \geq 0$ it holds that
    \[
        \Prob[S \geq (1+\delta)\mu]
        \leq e^{-\delta^2 \mu/(2+\delta)}.
    \]
    Let $x\coloneq \frac{\mu}{u^2m}\in(0,1]$. Then,
    \[
        \delta = \frac{2-x}{x} \quad \text{and} \quad x + \delta x = 2 \text{, hence} \quad \frac{\delta^2\mu}{2+\delta}=u^2m \frac{(2-x)^2}{2+x}.
    \]
    The claim follows by noticing that, for $0\leq x\leq 1$,
    \[
        \frac{(2-x)^2}{2+x}\geq \frac{1}{3}. \qedhere
    \]
\end{proof}

\section{Proofs for complexity-theoretic applications}

\subsection{\texorpdfstring{Constant $\pcNULL$ reduces to $\pcSINGULAR$}{Constant pc-NULL reduces to pc-SINGULAR}}
Here we provide a logspace many-one reduction from $\pcNULL_k$ for $k \in O(1)$ to $\pcSINGULAR$.
The underlying idea is due to Jain, Kerenidis, Kuperberg, Santha, Sattath and Zhang \cite{JKKSSZ09}. For a background on antisymmetric subspaces, we refer to \cite[Section 1.5]{Bhatia1997MatrixAnalysis}.

\begin{theorem}\label{thm:pcnull-reduces-to-pcsingular-for-constant-k}
    Let $k=O(1)$. Then
    \[
        \kappa\text{-}\mathsf{conditioned}\text{-}\mathsf{NULL}_k
        \preccurlyeq_m^{\mathsf L}
        (k\kappa^2)\text{-}\mathsf{conditioned}\text{-}\mathsf{SINGULAR}.
    \]
    In particular,
    \[
        \mathsf{pc}\text{-}\mathsf{NULL}_k
        \preccurlyeq_m^{\mathsf L}
        \mathsf{pc}\text{-}\mathsf{SINGULAR}.
    \]
\end{theorem}
\begin{proof}
    Let $A\in\C^{n\times n}$ be an input matrix to the problem $\pcNULL_k$ with $\norm{A}_{\op}\leq1$, and set
    \[
        H:=A^*A.
    \]
    Then $H$ is positive semidefinite, $\ker(H)=\ker(A)$, $\|H\|_{\op}\leq1$ and the eigenvalues of $H$ are the squared singular values of $A$.
    Write $0\leq\lambda_1\leq\cdots\leq\lambda_n$ for these eigenvalues, with corresponding orthonormal eigenvectors $\psi_1,\dots,\psi_n$.

    We use the $k$-th antisymmetric tensor power $\bigwedge^k\C^n$.
    This is a $\binom nk$-dimensional space with orthonormal basis
    \[
        e_{i_1}\wedge\cdots\wedge e_{i_k},
        \qquad
        1\leq i_1<\cdots<i_k\leq n.
    \]
    If we index the standard basis by the $k$-subsets of $[n]$, then we can identify $\bigwedge^k\C^n$ with $\C^{\binom nk}$.
    Alternatively, we can identify $\bigwedge^k\C^n$ with the antisymmetric subspace of $(\C^n)^{\otimes k}$ by letting
    \[
        v_1\wedge\cdots\wedge v_k
        := \frac1{\sqrt{k!}} \sum_{\sigma\in S_k} \mathrm{sgn}(\sigma)\, v_{\sigma(1)}\otimes\cdots\otimes v_{\sigma(k)}.
    \]

    Now consider the operator on $(\C^n)^{\otimes k}$
    \[
        H^{(k)}
        := \sum_{r=1}^k I^{\otimes(r-1)} \otimes H\otimes I^{\otimes(k-r)}.
    \]
    Since $H^{(k)}$ commutes with every permutation of the tensor factors, it maps antisymmetric tensors to antisymmetric tensors.
    Hence $\bigwedge^k\C^n$ as the antisymmetric subspace is invariant under $H^{(k)}$.

    Let $K\in\C^{\binom nk\times\binom nk}$ be the matrix representation of
    \[
        \frac1k H^{(k)}\big|_{\bigwedge^k\C^n}
    \]
    with respect to the basis $\{e_{i_1}\wedge\cdots\wedge e_{i_k}\}_{i_1<\cdots<i_k}$.
    Since $\|H\|_{\op}\leq1$, we have $\|H^{(k)}\|_{\op}\leq k$, and therefore $\|K\|_{\op}\leq1$.

    The vectors
    \[
        \psi_{i_1}\wedge\cdots\wedge\psi_{i_k},
        \qquad
        1\leq i_1<\cdots<i_k\leq n,
    \]
    form an orthonormal basis of $\bigwedge^k\C^n$ and satisfy
    \[
        H^{(k)}
        (\psi_{i_1}\wedge\cdots\wedge\psi_{i_k})
        =
        (\lambda_{i_1}+\cdots+\lambda_{i_k})
        (\psi_{i_1}\wedge\cdots\wedge\psi_{i_k}).
    \]
    Consequently, the spectrum of $K$ is
    \[
        \left\{ \frac{\lambda_{i_1}+\cdots+\lambda_{i_k}}{k} : 1\leq i_1<\cdots<i_k\leq n \right\}.
    \]

    Suppose first that $\dim\ker(A)\geq k$.
    Then $\lambda_1=\cdots=\lambda_k=0$, so $0$ is an eigenvalue of $K$ and hence $K$ is singular.

    Conversely, suppose that $\dim\ker(A)<k$ and $\sigma_k(A)\geq 1/\kappa$.
    Since the positive eigenvalues of $H=A^*A$ are the squares of the positive singular values of $A$, $\lambda_k\geq 1/\kappa^2$.
    Therefore
    \[
        \lambda_{\min}(K)
        = \frac{\lambda_1+\cdots+\lambda_k}{k}
        \geq \frac{\lambda_k}{k}
        \geq \frac1{k\kappa^2}.
    \]
    Thus $K$ is a valid no-instance of $(k\kappa^2)$-$\mathsf{conditioned}$-$\mathsf{SINGULAR}$.

    It remains to verify that $K$ can be output in logarithmic space.
    Since $k=O(1)$,
    \[
        \dim K=\binom nk=\poly(n).
    \]
    Index the rows and columns of $K$ by $k$-subsets
    \[
        I=\{i_1<\cdots<i_k\},
        \qquad
        J=\{j_1<\cdots<j_k\}
    \]
    of $[n]$.
    To determine the entries of $K$, note that for arbitrary $v_1,\dots,v_k\in\C^n$,
    \[
        H^{(k)}(v_1\wedge\cdots\wedge v_k)
        =
        \sum_{r=1}^k v_1\wedge\cdots\wedge Hv_r\wedge\cdots\wedge v_k.
    \]
    Hence
    \[
        H^{(k)}(e_{j_1}\wedge\cdots\wedge e_{j_k})
        =
        \sum_{r=1}^k\sum_{p=1}^n
        H(p,j_r)\,e_{j_1}\wedge\cdots\wedge e_p\wedge\cdots\wedge e_{j_k}.
    \]
    Thus an application of $H^{(k)}$ can replace at most one element of $J$, and therefore
    \[
        K(I,J)=0
    \]
    whenever $I$ and $J$ differ in more than one element.

    If $I=J$, only the terms with $p=i_r$ contribute, and hence
    \[
        K(I,I)
        =
        \frac1k\sum_{r=1}^k H(i_r,i_r).
    \]
    If $I$ and $J$ differ in exactly one element, say
    \[
        I\setminus J=\{i_a\},
        \qquad
        J\setminus I=\{j_b\},
    \]
    then the unique contributing term replaces $e_{j_b}$ by $e_{i_a}$.
    Reordering the resulting wedge into increasing order contributes the sign $(-1)^{a+b}$, so
    \[
        K(I,J)
        =
        \frac{(-1)^{a+b}}{k}H(i_a,j_b).
    \]

    Finally,
    \[
        H(i,j) = \sum_{\ell=1}^n \overline{A(\ell,i)}A(\ell,j).
    \]
    Since $k=O(1)$, the $k$-subsets $I,J$ can be stored using $O(\log n)$ space, and every entry of $H$, and hence every entry of $K$, can be computed using logarithmic workspace and standard logspace arithmetic. Therefore $A\mapsto K$ is a logspace many-one reduction.
\end{proof}

\subsection{Completeness of well-conditioned singularity}
\label{sec:completeness-of-conditioned-singularity}
In this subsection we prove the natural well-conditioned singularity completeness results for $\mathsf{Q_U MAL_1}, \mathsf{FewQ_U MAL_1}$ and $\mathsf{UQ_U MAL_1}$.
Let us recall the theorem:
\conditionedsingularitycompleteness*
\begin{proof}
    We establish hardness and containment for all three claims simultaneously.

    \textbf{The circuit Hamiltonian.}
    Let $x$ have length $n$, and let $Q_x=U_T\cdots U_1$ be the corresponding verifier, acting on $a+w=O(\log n)$ qubits, with $T=\operatorname{poly}(n)$.
    Set $V_0=I$, $V_t=U_t\cdots U_1$.
    Let $\Pi_{\mathrm{acc}}=\ket{1}\bra{1}_1 \otimes I$ project onto the first auxiliary qubit being in state $\ket{1}$, i.e. the case where the circuit accepts, and put $\Pi_{\mathrm{rej}}=I-\Pi_{\mathrm{acc}} = \ket{0}\bra{0}_1\otimes I$.
    The acceptance operator is
    \[
        E_x=(\bra{0^a} \otimes I)V_T^*\Pi_{\mathrm{acc}}V_T(\ket{0^a}\otimes I).
    \]
    Thus $\langle\psi|E_x|\psi\rangle$ is the acceptance probability of the circuit given the witness state $\ket{\psi}$.
    In particular, the perfectly accepting witness subspace is
    \[
        W_x=\ker(I-E_x)=\ker(\Pi_{\mathrm{rej}}V_T(\ket{0^a}\otimes I)).
    \]

    We use the clock construction underlying \cite[Lemma~35]{FeffermanRemscrim2021}, originating in \cite{KitaevShenVyalyi2002}.
    On $\mathbb{C}^{2^{a+w}}\otimes\mathbb{C}^{T+1}$, define
    \[
        \begin{aligned}
            H_{\mathrm{prop}}
            &=\frac12\sum_{t=1}^{T}
            \Bigl(
            I\otimes(\ket{t}\bra{t}+\ket{t-1}\bra{t-1})
            -U_t\otimes\ket{t}\bra{t-1}
            -U_t^*\otimes\ket{t-1}\bra{t}
            \Bigr),\\
            H_{\mathrm{in}}
            &=\left(\sum_{j=1}^{a}
            \ket{1}\bra{1}_j\otimes I_{\mathrm{rest}}\right)
            \otimes\ket{0}\bra{0},\\
            H_{\mathrm{out}}
            &=\Pi_{\mathrm{rej}}\otimes\ket{T}\bra{T},\\
            H_x&=H_{\mathrm{prop}}+H_{\mathrm{in}}+H_{\mathrm{out}}.
        \end{aligned}
    \]
    Here $\ket{1}\bra{1}_j$ acts on the $j$th auxiliary qubit.
    Each summand in $H_x$ is positive semi-definite, and one easily finds that,
    \[
        \|H_{\mathrm{prop}}\|_{\mathrm{op}}\leq2,
        \qquad
        \|H_{\mathrm{in}}\|_{\mathrm{op}}\leq a,
        \qquad
        \|H_{\mathrm{out}}\|_{\mathrm{op}}\leq 1,
        \qquad
        \text{and thus}
        \qquad
        \|H_x\|_{\mathrm{op}}\leq a+3.
    \]
    As a hard instance, we output
    \[
        A_x=\frac{H_x}{a+3}.
    \]
    This is a positive semidefinite $D\times D$ matrix, with $D=2^{a+w}(T+1)$ and operator norm $\|A_x\|_\op \leq 1$.
    Its entries are computable in deterministic logspace, since they involve only clock indices, auxiliary-bit tests, and entries of a single gate $U_t$.

    \textbf{Preservation of witness dimension:}
    Define the history-state isometry
    \[
        F_x\ket{\psi}
        =\frac{1}{\sqrt{T+1}} \sum_{t=0}^{T} V_t\ket{0^a}\ket{\psi}\otimes\ket{t}.
    \]
    For a vector $\sum_{t=0}^T\ket{y_t}\otimes\ket{t}$, zero propagation energy is equivalent to $\ket{y_t}=U_t\ket{y_{t-1}}$ for every $t$.
    Zero initialization energy requires $\ket{y_0} = \ket{0^a} \ket{\psi}$ for some $\ket{\psi} \in \mathbb C^{2^{w}}$, and zero output energy requires $\Pi_{\mathrm{rej}}\ket{y_T}=0$.
    Since all three terms are positive semi-definite, this proves
    \[
        \ker(A_x)=\ker(H_x)=F_xW_x,
        \qquad
        \dim\ker(A_x)=\dim W_x.
    \]
    Thus the construction preserves the dimension of the perfectly accepting witness subspace exactly.

    \textbf{The spectral gap analysis:}
    Assume that, for some $0<\gamma\leq1$,
    \[
        \Pr[Q_x\text{ accepts }\ket{\psi}]\leq1-\gamma
        \quad\text{for every normalized }\ket{\psi}\in W_x^\perp.
    \]
    We claim that
    \begin{equation}
        \label{eq:eigenvalue-lower-bound-ax}
        \lambda_{\min}^+(A_x)
        \geq \frac{\gamma}{24(a+3)(T+1)^2}.
    \end{equation}
    Since $A_x=H_x/(a+3)$, we have $\lambda_{\min}^+(A_x)=\lambda_{\min}^+(H_x)/(a+3)$.
    It therefore suffices to prove the corresponding lower bound for $H_x$, that is $\lambda_{\min}^+(H_x)\geq \frac{\gamma}{24(T+1)^2}$.

    To this end, fix an arbitrary vector
    \[
        \ket{y}=\sum_{t=0}^T\ket{y_t}\otimes\ket{t}\perp\ker(H_x),
        \qquad
        \ket{\psi_0}:=(\bra{0^a}\otimes I)\ket{y_0}.
    \]
    The vectors $\ket{y_t}$ and $\ket{\psi_0}$ need not be normalized.
    Writing $\Pi_j$ for the projector onto auxiliary qubit $j$ being $1$, the energy is
    \begin{equation}
        \label{eq:eta}
        \eta := \bra{y}H_x\ket{y}
        = \frac12\sum_{t=1}^T\norm{\ket{y_t}-U_t\ket{y_{t-1}}}^2
        + \sum_{j=1}^a\norm{\Pi_j\ket{y_0}}^2
        + \norm{\Pi_{\mathrm{rej}}\ket{y_T}}^2.
    \end{equation}

    \begin{claim}\label{claim:three-eta-inequalities}
        We have
        \begin{enumerate}[label=(\roman*),nosep]
            \item $\norm{V_t^*\ket{y_t}-\ket{y_0}}^2\leq2t\eta$ for every $0\leq t\leq T$,
            \item $\norm{\ket{y_0}-\ket{0^a}\ket{\psi_0}}^2\leq\eta$, and
            \item $\norm{\Pi_{\mathrm{rej}}\ket{y_T}}^2\leq\eta$.
        \end{enumerate}
    \end{claim}
    \begin{proof}
        Parts $(ii)$ and $(iii)$ follow from \cref{eq:eta}, using
        $\sum_j\Pi_j\succeq I-\ket{0^a}\bra{0^a}\otimes I$ for part $(ii)$.
        For part $(i)$, the case $t=0$ is immediate. For $t\geq1$, telescoping and Cauchy-Schwarz give
        \begin{align*}
            \norm{V_t^*\ket{y_t}-\ket{y_0}}^2
            &=\norm{\sum_{r=1}^t V_r^*\bigl(\ket{y_r}-U_r\ket{y_{r-1}}\bigr)}^2
            \leq t\sum_{r=1}^t\norm{\ket{y_r}-U_r\ket{y_{r-1}}}^2
            \leq2t\eta.\qedhere
        \end{align*}
    \end{proof}

    We next bound the rejection amplitude of $\ket{\psi_0}$. Decompose it as
    \begin{align*}
        \Pi_{\mathrm{rej}}V_T\ket{0^a}\ket{\psi_0}
        &=\Pi_{\mathrm{rej}}V_T \bigl(\ket{0^a}\ket{\psi_0}-\ket{y_0}\bigr)
        + \Pi_{\mathrm{rej}}V_T \bigl(\ket{y_0}-V_T^*\ket{y_T}\bigr)
        + \Pi_{\mathrm{rej}}\ket{y_T}.
    \end{align*}
    Since $\Pi_{\mathrm{rej}}$ is a contraction and $V_T$ is unitary, Cauchy-Schwarz and \cref{claim:three-eta-inequalities} imply
    \begin{align*}
        \norm{\Pi_{\mathrm{rej}}V_T\ket{0^a}\ket{\psi_0}}^2
        &\leq3\Bigl(\norm{\ket{0^a}\ket{\psi_0}-\ket{y_0}}^2
        + \norm{\ket{y_0}-V_T^*\ket{y_T}}^2
        + \norm{\Pi_{\mathrm{rej}}\ket{y_T}}^2\Bigr)\\
        &\leq 3(\eta+2T\eta+\eta)
        = 6(T+1)\eta.
    \end{align*}
    The rejection map $\Pi_{\mathrm{rej}}V_T(\ket{0^a}\otimes I)$ vanishes on $W_x$.
    Applying the assumed rejection gap to $\Pi_{W_x^\perp}\ket{\psi_0}$, therefore gives
    \begin{equation}
        \label{eq:projection-psi0}
        \gamma\norm{\Pi_{W_x^\perp}\ket{\psi_0}}^2
        \leq \norm{\Pi_{\mathrm{rej}}V_T\ket{0^a}\ket{\psi_0}}^2
        \leq 6(T+1)\eta.
    \end{equation}

    Now define the corresponding zero-energy history
    \[
        \ket{h}:=\sum_{t=0}^T V_t\ket{0^a}\Pi_{W_x}\ket{\psi_0}\otimes\ket{t}
        = \sqrt{T+1}\,F_x\Pi_{W_x}\ket{\psi_0}.
    \]
    Since $\ker(H_x) = F_x W_x$, $\ket{h}\in\ker(H_x)$.
    Using clock orthogonality, unitarity, and Cauchy--Schwarz, we obtain
    \begin{align*}
        \norm{\ket{y}-\ket{h}}^2
        &= \sum_{t=0}^T\norm{V_t^*\ket{y_t} -\ket{0^a}\Pi_{W_x}\ket{\psi_0}}^2\\
        &\leq 3\sum_{t=0}^T\Bigl(\norm{V_t^*\ket{y_t}-\ket{y_0}}^2
        + \norm{\ket{y_0}-\ket{0^a}\ket{\psi_0}}^2
        + \norm{\Pi_{W_x^\perp}\ket{\psi_0}}^2\Bigr)\\
        &\leq 3\sum_{t=0}^T \left(2t+1+\frac{6(T+1)}{\gamma}\right)\eta
        \leq\frac{24(T+1)^2}{\gamma}\eta,
    \end{align*}
    where the last line uses \cref{claim:three-eta-inequalities,eq:projection-psi0} and $\gamma\leq1$.
    Finally, $\ket{y}\perp\ket{h}$, so
    \[
        \norm{\ket{y}}^2\leq\norm{\ket{y}-\ket{h}}^2
        \leq \frac{24(T+1)^2}{\gamma}\bra{y}H_x\ket{y}.
    \]
    Taking the minimum Rayleigh quotient over nonzero vectors orthogonal to
    $\ker(H_x)$ yields
    \[
        \lambda_{\min}^+(H_x)
        = \min_{0\neq\ket{y}\perp\ker(H_x)} \frac{\bra{y}H_x\ket{y}}{\norm{\ket{y}}^2}
        \geq \frac{\gamma}{24(T+1)^2}.
    \]
    Dividing by $a+3$ proves \cref{eq:eigenvalue-lower-bound-ax}.

    On no-instances of all three classes, every witness is accepted with probability at most $1/3$ and $W_x=\{0\}$, so the assumption holds with $\gamma=2/3$. The same is true on yes-instances of $\mathsf{FewQ_U MAL}_1$ and $\mathsf{UQ_U MAL}_1$, by their acceptance gap on $W_x^\perp$. In these cases,
    \[
        \lambda_{\min}^+(A_x)
        \geq \frac{1}{36(a+3)(T+1)^2}
        \geq \frac{1}{\poly(n)}.
    \]
    For ordinary $\mathsf{Q_U MAL}_1$, the assumption need not hold on yes-instances, only the existence of a zero eigenvalue is required.
    This establishes the three hardness claims.

    \textbf{Containment:}
    We use the one-sided kernel-verification procedure invoked in \cite[Lemma~35]{FeffermanRemscrim2021} (see also \cite[Lemma~19 and Proposition~32]{fefferman_lin}).
    For a Hermitian contraction $H$ and $\delta\geq1/\poly(n)$, there is a deterministic logspace-uniform unitary verifier, using $O(\log n)$ qubits and $T=\poly(n)$ gates, with the following properties:
    \begin{enumerate}[label=(\roman*)]
        \item every normalized vector in $\ker(H)$ is accepted with probability $1$, and
        \item every normalized vector in the span of the eigenvectors of $H$ with eigenvalues satisfying $|\lambda|\geq\delta$ is accepted with probability at most $1/3$.
    \end{enumerate}
    When $A$ is the input matrix, apply this procedure to $H=A^*A$ and $\delta=\kappa(n)^{-2}$.
    Since $\ker(H)=\ker(A)$ and the eigenvalues of $H$ are the squared singular values of $A$, this gives perfect completeness and soundness for standard $\pcSINGULAR$ (as is already described in \cite[Lemma~35]{FeffermanRemscrim2021}).
    For $\pcFewSINGULAR$ and $\pcUSINGULAR$, the positive singular value gap additionally ensures acceptance probability at most $1/3$ on $\ker(A)^\perp$ on yes-instances.
    Their kernel dimensions therefore give the required perfectly accepting subspaces, proving containment in $\mathsf{FewQ_U MAL}_1$ and $\mathsf{UQ_U MAL}_1$, respectively.
\end{proof}
\end{document}